\documentclass[10pt,a4paper]{amsart}
\usepackage[cp1250]{inputenc}
\usepackage[english]{babel}
\usepackage{color}
\usepackage[colorlinks,linkcolor=blue,citecolor=red]{hyperref}
\usepackage{amssymb}
\usepackage{amsmath}
\usepackage{amsfonts}
\usepackage[IL2]{fontenc}
\usepackage{url}
\usepackage{graphicx}
\usepackage{standalone}

\usepackage{pgf,tikz}
\usepackage{mathrsfs}
\usetikzlibrary{arrows}
\usetikzlibrary{snakes}
\usetikzlibrary{arrows}
\usetikzlibrary{calc}
\usetikzlibrary{patterns}
\usepackage{tkz-fct}

\newtheorem{defi}{\bf D\scriptsize EFINITION \normalsize}
\newtheorem{theorem}{\bf T\scriptsize HEOREM \normalsize}
\newtheorem{lm}{\bf L\scriptsize EMMA \normalsize}
\newtheorem{dk}{\bf C\scriptsize OROLLARY \normalsize}
\newtheorem{rem}{\bf R\scriptsize EMARK \normalsize}
\newtheorem{exa}{\bf E\scriptsize XAMPLE \normalsize}
\newtheorem{pro}{\bf P\scriptsize ROBLEM \normalsize}
\newtheorem{prop}{\bf P\scriptsize ROPOSITION \normalsize}
\newtheorem{no}{\bf N\scriptsize OTE \normalsize}

\newenvironment{remark}{\begin{rem}\rm}{\end{rem}}
\def\kopr{\hfill\raisebox{3pt}{\framebox{$\star$}}}
\newenvironment{example}{\begin{exa}\rm}{$\kopr$\end{exa}}
\newenvironment{definition}{\begin{defi}\rm}{\end{defi}}

\newenvironment{lemma}{\begin{lm}\it}{\end{lm}}
\newenvironment{corollary}{\begin{dk}\it}{\end{dk}}
\newenvironment{proposition}{\begin{prop}\it}{\end{prop}}

\newcommand{\inte}[2]{\int \limits_{#1}^{#2}}

\newcommand{\nadsebou}[2]{\begin{array}{c} #1 \\ #2 \end{array}}
\newcommand{\hzav}[1]{\left[ #1 \right]}
\newcommand{\zav}[1]{\left( #1 \right)}

\newcommand{\abs}[1]{\left| #1 \right|}

\newcommand{\sgn}{{\rm sgn} }

\newcommand{\R}{{\mathbb{R}}}

\newcommand{\dd}{{\rm d}}
\newcommand{\ii}{{\rm i}}
\allowdisplaybreaks
\begin{document}
\title{Pedal coordinates, solar sail orbits, Dipole drive and other force problems}
 \author{Petr Blaschke}
 \thanks{}
\address{ Mathematical Institute, Silesian University in Opava, Na Rybnicku 1, 746 01 Opava, Czech Republic}

\email{Petr.Blaschke@math.slu.cz}
\begin{abstract} 
It was shown that pedal coordinates provides natural framework in which to study force problems of classical mechanics in the plane. A trajectory of a test particle under the influence of central and Lorentz-like forces can be translated into pedal coordinates at once without the need of solving any differential equation. We will generalize this result to cover more general  force laws and also show an advantage of pedal coordinates in certain variational problems.  These will enable us to link together many dynamical systems as well as problems of calculus of variation.
Finally -- as an illustrative example -- we will apply obtained results to compute orbits of Solar sail and Dipole drive.
\end{abstract}
\maketitle
\section{Introduction}\label{Intro}
Even when dealing with a differential equation that has solutions expressible in terms of well known functions (or even elementary functions), the complexity of the resulting formulas may be so high to yield surprisingly little information about the solution itself. In essence, there is a little difference between complicated analytical solution and a numerical one. On the other hand, a sufficiently simple analytical solution can provide great insight into the problem.

A rather revealing example of this sort is the following innocently looking exercise from calculus of variation:
\\

\emph{Find a curve of a given length fixed on both ends that sweeps maximal volume when rotated around the $x$-axis.}
\\ 

That is, we must find a function $y\equiv y(x)$ that maximizes
$$
L[y]:=\inte{-1}{1} \pi y^2{\rm d}x,
$$
under conditions:
$$
y(-1)=y(1)=0,\qquad \inte{-1}{1}\sqrt{1+{y^\prime}^2}{\rm d}x=l,
$$
where $l$ is a given length.

Since the functional does not depend on $x$ explicitly, we can use the Beltrami identity to reduce the order of Euler-Lagrange equation to one. (The problem actually appears so simple that -- some years ago -- I have set it to my undergraduate students, thinking this would be an easy exercise.) But, in fact, the solution is expressible only implicitly in term of Incomplete elliptic integrals:
$$
\sqrt{\frac{\lambda-C}{\pi}}F\zav{y\sqrt{\frac{\pi}{C+\lambda}},k}-\sqrt{\frac{\lambda-C}{\pi}}E\zav{y\sqrt{\frac{\pi}{C+\lambda}},k}+Cy=x-d.
$$

The resulting curve is known as ``the elastic curve'', since it represent the shape of an elastic rod which bends as its two ends are brought close together.

The formula is not very long and the special functions involved are well known, but can any one guess from it, how the curve actually looks like?

A much more advantageous approach to this problem is \textit{not to solve it at all.} In particular, one must avoid the temptation to use the Beltrami identity, since the second order Euler-Lagrange equation, \textit{left as it is}, provides much better description of the resulting curve:
$$
2\pi y=\lambda \frac{y''}{\zav{1+{y'}^2}^\frac32},
$$
where $\lambda$ is the Lagrange multiplier.

This can be readily interpreted since the formula on the right hand side is nothing else than the curve's (oriented) curvature $\kappa$. Thus the equation above tells us that the solution's curvature is proportional to $y$: 
$$
 \kappa\propto y.
$$ 

Hence, we can easily see, that the curve must start on the $x$-axis ($y=0$) flat as a line and curves more and more as it moves away from it up to some critical distance when the curvature is so large that the curve is turned back and in perfect symmetry moves towards the $x$ axis flatter and flatter as she goes until she hits the $x$-axis again flat as a line. In other words, it must look like this:
\begin{center}
\includegraphics[scale=0.5,clip, trim=0.5cm 14cm 5.5cm 1cm]{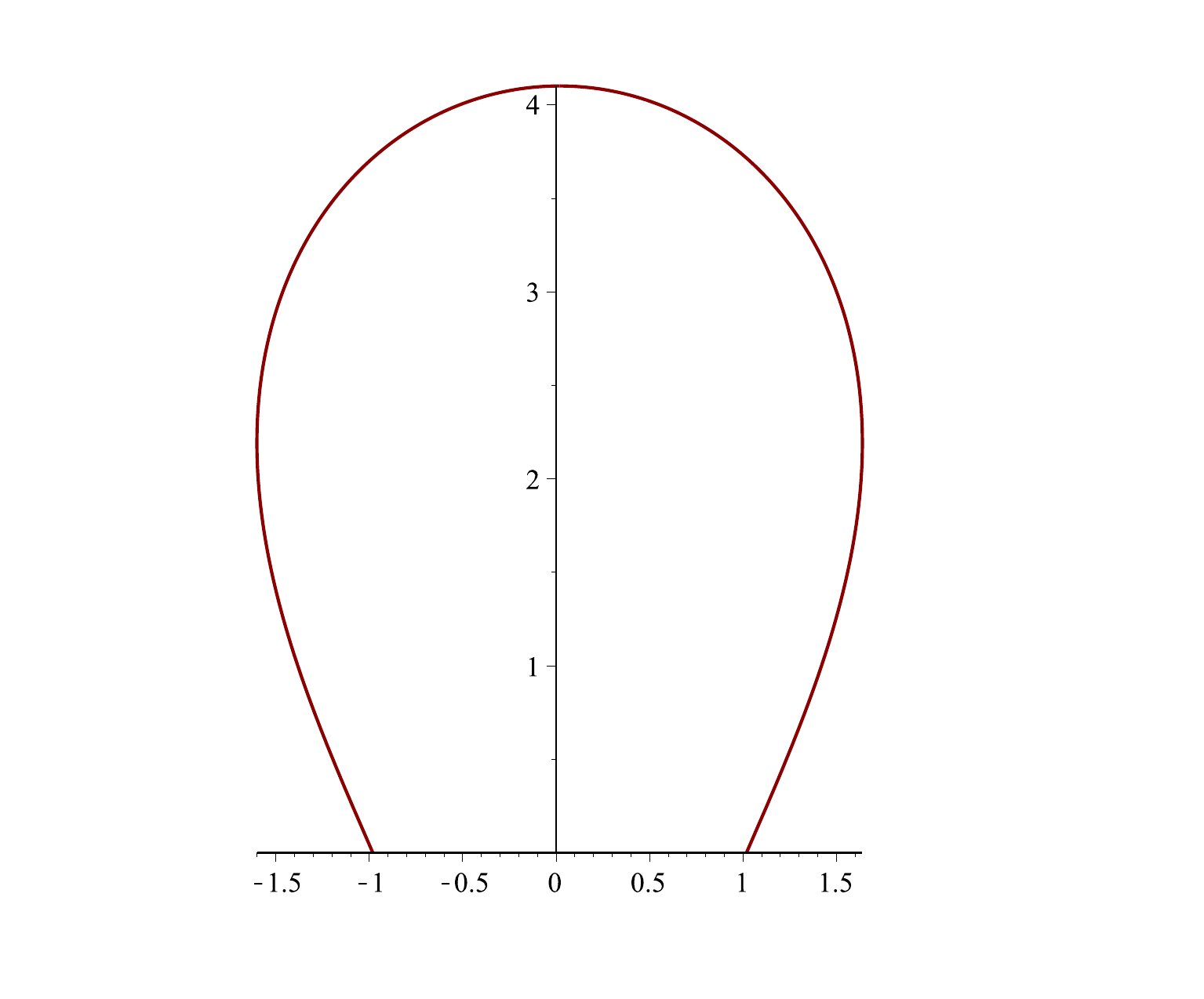}
\end{center}
(The curve is reminiscent to a tennis racket. Perhaps, during a swing, it is also somehow advantageous to maximize the volume of air, the racket covers.)
\bigskip

This example suggests that instead of focusing on solutions, much can be achieved by choosing the ``right'' coordinates -- $(\kappa, y)$ in this case.

In \cite{Blaschke6} it was shown that similarly ``right'' coordinates for solving certain  force problems in a plane, i.e. dynamical systems of the form
\begin{equation*}
{\bf \ddot x}=F({\bf x},{\bf \dot x}),\qquad {\bf x}\in\mathbb{R}^2,
\end{equation*}
are so-called \textit{pedal coordinates} (\cite{Yates,Edwards}), where the position of a point ${\bf x}$ on a given curve $\gamma$ is determined by two numbers $(r,p)$:  the distance $r$ of ${\bf x}$ from a given point (the so-called pedal point) and the distance $p$ of the \textit{tangent} of $\gamma$ at ${\bf x}$.

More precisely, the following holds:
\begin{theorem}\label{T1}(From \cite{Blaschke6}) Consider a dynamical system:
\begin{equation}\label{dynsys}
{\bf \ddot x}=F^\prime\zav{r^2}{\bf x}+2 G^\prime\zav{r^2}{\bf \dot x}^\perp,\qquad r:=\abs{{\bf x}},\qquad {\bf \dot x}^\perp\cdot {\bf \dot x}=0,
\end{equation}
describing an evolution of a test particle (with position ${\bf x}$ and velocity ${\bf \dot x}$) in the plane in the presence of central $F$ and Lorentz like $G$ potential. The quantities:
$$
L={\bf x}\cdot {\bf \dot x}^\perp+G\zav{r^2}, \qquad c=\abs{{\bf \dot x}}^2-F\zav{r^2},
$$
are conserved in this system.

Then the curve traced by ${\bf x}$ is given in pedal coordinates by
\begin{equation}
\frac{\zav{L-G(r^2)}^2}{p^2}=F(r^2)+c,
\end{equation}
with the pedal point at the origin.
\end{theorem}

In this paper we will generalize this result in multiple ways to include wider classes of force problems: 

\begin{theorem}\label{T2}
Let $f\equiv f(r,p),g\equiv g(r,p)$ be functions of pedal coordinates $r,p$, 
$$
r:=\abs{{\bf x}},\qquad p:=\frac{{\bf \dot x}\cdot {\bf x}^\perp}{\abs{{\bf \dot x}}}, \qquad {\bf x}\in\mathbb{R}^2,
$$ where the symbol ${\bf a}^\perp$ denotes a vector perpendicular to a two dimensional vector ${\bf a}$, i.e. $(a_1,a_2)^\perp:=(-a_2,a_1)$.

Then any solution of the following dynamical systems (on the left) can be written in pedal coordinates (on the right):
\begin{align}
\label{T2F1}\ddot {\bf x}&=f {\bf x}+g \dot {\bf x}^\perp, & \frac{\zav{\int g r{\rm d}r}^2}{p^2}&=2\int f r{\rm d}r, &\\
\label{T2F2}\ddot {\bf x}&=f  \frac{{\bf x}}{\abs{\dot {\bf x}}^\alpha}+g \frac{\dot {\bf x}^\perp}{\abs{\dot {\bf x}}^\beta}, &  \zav{\frac{(1+\beta)\int g p^\beta r{\rm d}r}{p^{1+\beta}}}^{2+\alpha}&=\zav{(2+\alpha)\int f r{\rm d}r}^{1+\beta}, & \nadsebou{\alpha\not= -2}{\beta\not=-1}\\
\label{T2F3}\ddot {\bf x}&=f  \frac{{\bf x}}{\abs{\dot {\bf x}}^{\beta-1}}+g \frac{\dot {\bf x}^\perp}{\abs{\dot {\bf x}}^\beta}, & \frac{\int g p^\beta r{\rm d}r}{p^{1+\beta}}&=\int f r{\rm d}r, & \beta\not=-1.\\
\label{T2F4}\ddot {\bf x}&=f  \frac{{\bf x}}{\abs{\dot {\bf x}}^{\alpha}}+g \abs{\dot {\bf x}} \dot {\bf x}^\perp, & \frac{e^{(\alpha+2)\int \frac{g}{p} r{\rm d}r}}{p^{\alpha+2}}&=(\alpha+2)\int f r{\rm d}r, & \alpha\not=-2.\\
\label{T2F5}\ddot {\bf x}&=f  \abs{\dot {\bf x}}^{2} {\bf x}+g \frac{\dot {\bf x}^\perp}{\abs{\dot {\bf x}}^{\beta}}, & \frac{(\beta+1)\int gp^\beta r{\rm d}r}{p^{\beta+1}}&=e^{(\beta+1)\int f r{\rm d}r}, & \beta\not=-1.\\
\label{T2F6}\ddot {\bf x}&=f  \abs{\dot {\bf x}}^{2} {\bf x}+g \abs{\dot {\bf x}} \dot {\bf x}^\perp, & \int \frac{g}{p} r{\rm d}r - \ln p&=\int f r{\rm d}r, & \\
\label{T2F7}\ddot {\bf x}&=f(r,p,\abs{\bf \dot x}) {\bf x}, & \frac{L^2}{p^2}&=2\int r f\zav{r,p,\frac{L}{p}}{\rm d}r,\qquad & L:={\bf \dot x}\cdot {\bf x}^\perp.
\end{align}
\end{theorem}
\begin{remark}
The pedal expressions above corresponding to various dynamical systems are no longer algebraic in nature but contain indefinite integrals making them technically differential equations. But this notation was chosen in pursuit of full generality. Often one encounters problems with $f,g$ that are easily integrable. Particularly, if one demand that $f,g$ are functions of $r$ only a purely algebraic generalization of Theorem \ref{T1} can be generated, for example, as follows:
\end{remark}
\begin{corollary}\label{C1}
In the notation of Theorem \ref{T2}:
\begin{align*}
\ddot {\bf x}&=F'(r^2)  \frac{{\bf x}}{\abs{\dot {\bf x}}^\alpha}+2 G'(r^2) \dot {\bf x}^\perp, & \zav{\frac{G(r^2)+L}{p}}^{2+\alpha}&=\frac{2+\alpha}{2}F(r^2)+c, \\ 
\ddot {\bf x}&=\frac{F'(r^2)}{G'(\abs{{\bf \dot x}}^2)} {\bf x}, & G\zav{\frac{L^2}{p^2}}&=F(r^2)+c, 
\end{align*}
\end{corollary}
And so on.
\begin{remark}
Theorem \ref{T2} augments Theorem \ref{T1} primarily by the fact that right hand sides of equations of motion can explicitly depend on $\abs{{\bf \dot x }}$. But the added benefit is actually bigger than that. The reader might perhaps view as a disadvantage that the force must always be composed from a central component (i.e. that points in direction of ${\bf x}$) and a Lorentz-like component (that points in direction of ${\bf \dot x}^\perp$). 

But ${\bf x},{\bf \dot x}^\perp$ are just two vectors in a plane, typically independent and thus constitute a base. Other vectors that might be of interest like ${\bf x}^\perp, {\bf \dot x}$ can be hence written as their linear combination. In fact, it is simple to verify that: 
\begin{align*} 
{\bf x}^\perp&=\frac{p}{p_c} {\bf x}+\frac{r^2}{p_c \abs{\dot {\bf x}}}{\dot {\bf x}}^\perp, &
\dot {\bf x}&=\frac{\abs{\dot {\bf x}}}{p_c} {\bf x}+\frac{p}{p_c}{\dot {\bf x}}^\perp,\qquad p_c:=\frac{{\bf x}\cdot {\bf \dot x}}{\abs{{\bf \dot x}}},
\end{align*}
where $p_c$ is the (oriented) distance of the normal to the origin which is also called ``contrapedal coordinate''. The connection to the ordinary pedal coordinates is: 
$$
p_c^2=r^2-p^2.
$$

The fact that we can handle the quantity $\abs{{\bf \dot x }}$ makes this change of basis practical. Later we will give several examples which exploits this feature. 
\end{remark}

We will also show that using pedal coordinates, we can also easily solve some variational problems:
\begin{proposition}\label{P1}
Any extremal curve of the functional:
$$
\mathcal{L}\hzav{r}:=\inte{s_0}{s_1}f(r)\dd s,
$$
where
$$
\dd s:=\sqrt{{r'_\varphi}^2+r^2}\dd \varphi=\sqrt{1+{y'_x}^2}\dd x,
$$
is the arc-length measure, has pedal equation
$$
\frac{L}{p}=f(r).
$$
\end{proposition}
\begin{remark}
There are, of course, many more functionals whose Euler-Lagrange equation can be solved in pedal coordinates. For instance, those of the form
$$
\mathcal{L}\hzav{r}:=\inte{\varphi_0}{\varphi_1}f(r,r'_\varphi, r''_\varphi,\dots, r^{(n)}_{\varphi})\dd \varphi.
$$
As can be seen from Proposition \ref{P2}, which we will state later.

We choose this particular form since it is more suitable for \textit{interpreting} pedal equation and it will enable us to even make connections between variational problems and force problems and vice versa. As we will see.
\end{remark}
%
%

To show that generalizations given in Theorem \ref{T2} and Proposition \ref{P1} are useful, we will give several examples of ``natural'' dynamical systems and variational problems that have some intrinsic interest other than their solution can be given in pedal coordinates. We will focus mainly on problems of celestial mechanics that will include trajectories of a Solar sail and Robert Zubrin's ``Dipole Drive'' \cite{Zubrin}. This will be done in Sections 5 and 6.

In Section 2 we will give a fairly detailed introduction of pedal coordinates since  -- although it is quite an old subject -- it has been largely forgotten. This introduction differs substantially from the one given in \cite{Blaschke6}, since there the quantities $p,p_c$ were always nonnegative (as proper distances should). But in this paper we are treating $p,p_c$ as \textit{oriented} distances that can be negative -- which simplifies some formulas, as we will see.

In Section 3 we will prove Theorem \ref{T2} and Proposition \ref{P1}.

In Section 4 we will give several examples of variational problems that are easily solvable in pedal coordinates, including ``central Brachistochrone'' and ``dark Catenary'' problems. 

Finally, in addition to pedal coordinates $(p,r)$ there are of course, number of other interesting coordinate systems -- e.g. \textit{intrinsic coordinates} where arc-length $s$ and curvature $\kappa$ are specified, and many more. We will attempt to systematize the analysis introducing so-called ``non-local variables''. This will enable us to ``solve'' number of dynamical system that are beyond the reach of even pedal coordinates, including ``aerobraking'' or trajectories of a solar sail pointing ``prograde'' -- i.e. in the direction of movement.

\section{Pedal coordinates}
Let us introduce for any point ${\bf x}\in\mathbb{R}^2$ on a plane curve $\gamma$ the following symbols: 
\begin{align*}
{\bf x}&:=(x,y), & {\bf x}^\perp&:=(-y,x), & {\bf x}\cdot {\bf y}^\perp&=-{\bf x}^\perp \cdot {\bf y}, & \zav{{\bf x}^\perp}^\perp&=-{\bf x}. \\
p&:=\frac{{\bf x}^\perp \cdot {\dot {\bf x}}}{\abs{\dot {\bf x}}}, &  p_c&:=\frac{{\bf x} \cdot {\dot {\bf x}}}{\abs{\dot {\bf x}}}, & \kappa&:=\frac{{\dot {\bf x}}^\perp \cdot {\ddot {\bf x}} }{\abs{\dot {\bf x}}^3}, & p^2+p_c^2&=r^2.\\
\end{align*}

This diagram roughly illustrates the definition of pedal coordinates $p,p_c,r$:
\begin{center}
{\small
\begin{tikzpicture}[domain=-4:5]
\draw [thick,color=blue] (-4,3) to [out=0,in=180] (0,2) to [out=0, in=225] (2,3);
\coordinate [label=right:{$\gamma$}] (M1) at (2,3);
\coordinate [label=above:{${\bf x}$}] (M2) at (0,2);
\draw [thin] (-4,2) to (2,2);
\draw [dashed] (-3,2) to (-3,0);
\draw [dashed] (0,2) to (-3,0);
\filldraw (-3,0) circle (0.05);
\filldraw (0,2) circle (0.05);
\filldraw (-3,2) circle (0.05);
\coordinate [label=below:{$0$}] (M3) at (-3,0);
\coordinate [label=above:{$P({\bf x})$}] (M4) at (-3,2);
\draw [densely dotted] (0,2) to (0,0);
\draw [densely dotted] (-3,0) to (0,0);
\filldraw (0,0) circle (0.05);
\coordinate [label=left:{$p$}] (M6) at (-3,1);
\coordinate [label=below right:{$r$}] (M7) at (-1.5,1);
\coordinate [label=below right:{$p_c$}] (M8) at (-1.5,0);
\end{tikzpicture}
}
\end{center}

But with $p,p_c$ we must remember that these are \textit{signed} distances, which can be negative.

The point $P({\bf x})$ indicates a point on a \textit{new}, so-called ``pedal curve'' which we denote $P(\gamma)$ and to transformation that bring a curve to its pedal curve we will refer as ``pedal transform'' $P$.

In fact, pedal coordinates was originally devised (as the name suggests) to make pedal transformation easy to do. In Cartesian coordinates it is generally required to solve a differential equation. But same thing can be by done algebraically in pedal coordinates:

Specifically (see \cite[p. 228]{williamson}), to any curve given by the equation
$$
f(p,r, p_c)=0,
$$ 
in pedal coordinates, the pedal curve satisfies the equation
$$
f\zav{r,\frac{r^2}{p},\frac{r}{p}p_c}=0.
$$
(It is always advantegous to show how various transforms affects also the quantity $p_c$, since in some settings it is the more natural variable to work with. This does not mean, of course, that the function $f$ above has three variables. It has only two since $p_c$ is not independent.)

Some of the simplest curve are given in pedal coordinates as follows:
\begin{align*}
p&=a, &\text{line distant $a$.} \\
r&=a, &\text{point distant $a$.} \\
p_c&=a, &\text{involute of a circle.} \\
2pR^2&=r^2+R^2-|a|^2, &\text{circle with radius $R$ and center at $a$.}
\end{align*}
Involute of a circle is not usually regarded as a ``simple curve'' but in pedal coordinates it has extremely simple equation so perhaps it is. (For a derivation see \cite{Blaschke6}.)

Here is an geometric proof for the equation of a circle:
\begin{center}
\definecolor{qqzzqq}{rgb}{0.,0.6,0.}
\definecolor{uuuuuu}{rgb}{0.26666666666666666,0.26666666666666666,0.26666666666666666}
\definecolor{xdxdff}{rgb}{0.49019607843137253,0.49019607843137253,1.}
\definecolor{qqqqcc}{rgb}{0.,0.,0.8}
\definecolor{qqqqff}{rgb}{0.,0.,1.}
\begin{tikzpicture}[line cap=round,line join=round,>=triangle 45,x=1.0cm,y=1.0cm,scale=0.7]
\clip(3.02,2.24) rectangle (14.24,10.76);
\draw [line width=1.6pt,color=qqqqcc] (9.48,6.5) circle (3.1294727990509825cm);
\draw [line width=1.6pt,domain=3.02:14.24] plot(\x,{(-4.488628851203698-2.1309949286509227*\x)/-2.2918247345868394});
\draw [line width=0.8pt] (5.8,5.36)-- (4.8068690528611295,6.42808422616822);
\draw [line width=0.8pt,color=qqzzqq] (5.8,5.36)-- (8.342136018487947,7.723740508418621);
\draw [line width=0.8pt,color=qqzzqq] (5.8,5.36)-- (9.48,6.5);
\draw [line width=0.8pt,color=qqzzqq] (9.48,6.5)-- (8.342136018487947,7.723740508418621);
\draw [line width=0.8pt] (8.342136018487947,7.723740508418621)-- (7.349005071349078,8.79182473458684);
\begin{normalsize}
\draw [fill=qqqqff] (9.48,6.5) circle (2.5pt);
\draw[color=qqqqff] (9.94,5.91) node {$a$};
\draw [fill=xdxdff] (7.349005071349078,8.79182473458684) circle (2.5pt);
\draw[color=uuuuuu] (6.84,9.37) node {${\bf x}$};
\draw [fill=qqqqff] (5.8,5.36) circle (2.5pt);
\draw[color=qqqqff] (6.06,4.53) node {$0$};
\draw [fill=uuuuuu] (4.8068690528611295,6.42808422616822) circle (1.5pt);
\draw[color=black] (5.54,6.13) node {$p$};
\draw [fill=uuuuuu] (8.342136018487947,7.723740508418621) circle (1.5pt);
\draw[color=qqzzqq] (7.51,6.54) node {$p_c$};
\draw[color=qqzzqq] (8.28,5.57) node {$|a|$};
\draw[color=qqzzqq] (9.56,7.21) node {$R-p$};
\draw[color=black] (8.9,7.89) node {$R$};
\end{normalsize}
\end{tikzpicture}
\end{center}
It is just the Pythargorean theorem on the green triangle (and applying identity $p_c^2=r^2-p^2$).

Letting $R\to 0$ shows that the equation $r=|a|$ indeed belong to a point -- and not to a circle as is usual in polar coordinates. 

($r=|a|$ is a circle in polar coordinates since the other coordinate -- the angle $\varphi$ -- is arbitrary. In pedal coordinates it is the coordinate $p$ which is arbitrary. But  \emph{not completely} arbitrary. It has to always be true that
$$
\abs{p}\leq r,
$$
in other words the distance to the tangent at a point is always smaller than the distance to that point. This is consistent with the picture that a curve consisted of single point has an arbitrary tangent line at this point.)

\begin{remark}
At this point it should be clear that pedal coordinates does not tell us everything about the curve and they actually describes many curves at once -- if you choose to differentiate between them.

The equation $p=a$ is valid for \emph{any} line distant $a$ and $r=a$ for any point distant $a$, etc.

Obviously, the pedal coordinates do not care about rotation around the pedal point and about the curve's parametrization, but it is actually not easy to tell in general the nature of ambiguity associated to a pedal equation -- in fact, it differs from equation to equation. (For more information see \cite{Blaschke6}.)

This is actually \emph{an advantage} of pedal coordinates over other systems if you are interested only in the general shape of the curve and do not want to be distracted by details. 
\end{remark}

\begin{example}
Pedal coordinates are extremely dependent on the choice of origin. For instance, take pedal equation of an ellipse:
\begin{equation}\label{pedellipsefocus}
\frac{b^2}{p^2}=\frac{2a}{r}-1,
\end{equation}
where $a,b$ are semi-axis. This holds only when the origin is at one focus.

If the origin is at center of the ellipse, the equation changes to:
\begin{equation}\label{pedellipsecenter}
\frac{a^2 b^2}{p^2}=- r^2+a^2+b^2.
\end{equation}

If the origin lies on the ellipse it is possible to derive such an expression:

For: 
$$
\abs{{\bf x}-a}+\abs{{\bf x}-b}=a+b,
$$
it holds:
\begin{equation}\label{pedellipseboundary}
\frac{ab(a+b)^2\zav{a(a+b)+R}\zav{b(a+b)+R}r^4}{p^2}=\zav{a^2(a+3b)-r^2(a-b)}\zav{r^2(a-b)+b^2(b+3a)}\zav{R+2ab}^2,
\end{equation}
where
$$
R=\sqrt{r^2(a-b)^2+4a^2b^2}.
$$

For general position of the origin, we believe, that pedal equation cannot be written in elementary functions only but some special functions (specifically elliptic functions) are required.

(Since these claims bare no significance for the goals of the paper we just state them as a curiosity without proof.)
\end{example}
\begin{example}
Even more dramatic demonstration of this property is a curve known (by some unfortunate mistranslation from latin) by the name ``Witch of Agnesi'' but the modern reader would recognize it just as a graph of derivative of arctan$(x)$ function:
$$
y=\frac{1}{1+x^2}.
$$
Here is its picture together with a geometrical construction:
\begin{center}
\definecolor{xdxdff}{rgb}{0.49019607843137253,0.49019607843137253,1.}
\definecolor{uuuuuu}{rgb}{0.26666666666666666,0.26666666666666666,0.26666666666666666}
\definecolor{qqqqcc}{rgb}{0.,0.,0.8}
\definecolor{qqqqff}{rgb}{0.,0.,1.}
\begin{tikzpicture}[line cap=round,line join=round,>=triangle 45,x=1.0cm,y=1.0cm,scale=0.7]
\clip(-1.4286443374600932,-0.5879277621259713) rectangle (19.063698260384477,9.631017218021487);
\draw [line width=1.6pt,dash pattern=on 4pt off 4pt,color=qqqqcc] (8.98,4.66) circle (1.82cm);
\draw [line width=0.8pt,dash pattern=on 4pt off 4pt] (8.98,-0.5879277621259713) -- (8.98,9.631017218021487);
\draw [line width=0.4pt,domain=-1.4286443374600932:19.063698260384477] plot(\x,{(-11.7936-0.*\x)/-1.82});
\draw [line width=0.4pt,domain=8.98:19.063698260384477] plot(\x,{(-7.14226350200407--1.3515494042157759*\x)/1.7586796295259148});
\draw [line width=0.4pt] (13.716485274978751,-0.5879277621259713) -- (13.716485274978751,9.631017218021487);
\draw [line width=0.4pt,domain=-1.4286443374600932:19.063698260384477] plot(\x,{(--18.765031197669227-0.*\x)/4.476872246524337});
\draw[line width=1.6pt,color=qqqqcc] (13.716485274978753,4.191549404215776)(5.339999963600002,4.6599999818) -- (5.339999927199999,4.6599999636) -- (5.3399998544,4.6599999272) -- (5.33999970879999,4.6599998544) -- (5.339999417599958,4.6599997088) -- (5.339998835199816,4.6599994176) -- (5.33999767039926,4.659998835200001) -- (5.3399953407970235,4.659997670400001) -- (5.339990681588079,4.6599953408) -- (5.339981363152296,4.6599906816) -- (5.339962726209164,4.6599813632) -- (5.339925452036629,4.6599627264000025) -- (5.339850902546467,4.65992545280002) -- (5.339701798985526,4.659850905600167) -- (5.339403573539414,4.659701811201334) -- (5.338807049336295,4.659403622410673) -- (5.337613707574297,4.658807244885382) -- (5.335225849729389,4.657614490283046) -- (5.330445429561896,4.655228984664365) -- (5.320865713586269,4.65045800211488) -- (5.301630314376035,4.640916266517601) -- (5.262851883031912,4.6218346312946625) -- (5.18403361668904,4.583686047280813) -- (5.103472464095049,4.545571025268226) -- (5.021091651408985,4.507506327806497) -- (4.936810077819049,4.469508695313047) -- (4.850542000318019,4.431594838710913) -- (4.76219669065933,4.393781432079438) -- (4.671678061582302,4.356085105321228) -- (4.625572160042566,4.337286029859293) -- (4.578884259041076,4.318522436848505) -- (4.531600652695354,4.299796389348833) -- (4.483707216769672,4.281109946292107) -- (4.435189392400174,4.262465162255643) -- (4.386032169056435,4.243864087236345) -- (4.3362200666973045,4.2253087664252975) -- (4.285737117076185,4.206801239982909) -- (4.234566844148043,4.188343542814585) -- (4.182692243527335,4.169937704346998) -- (4.130095760942797,4.151585748304945) -- (4.076759269631314,4.1332896924888445) -- (4.022664046609607,4.115051548552875) -- (3.967790747757879,4.096873321783798) -- (3.9121193816456876,4.078757010880467) -- (3.855629282025117,4.060704607734085) -- (3.7982990789114,4.042718097209175) -- (3.7401066681657538,4.024799456925369) -- (3.681029179488792,4.006950657039948) -- (3.6210429427271347,3.989173660031233) -- (3.5601234523881127,3.9714704204828046) -- (3.4982453302507466,3.9538428848686027) -- (3.4353822859524747,3.9362929913389064) -- (3.371507075422712,3.918822669507235) -- (3.3065914570245623,3.9014338402381883) -- (3.240606145256101,3.8841284154362454) -- (3.1735207618510186,3.8669082978355513) -- (3.105303784106681,3.849775380790711) -- (3.0359224902544364,3.832731548068617) -- (2.9653429016726687,3.815778673641323) -- (2.893529721727578,3.798918621480005) -- (2.8204462710098586,3.7821532453500195) -- (2.746054418716974,3.765484388607075) -- (2.670314509910674,3.7489138839945597) -- (2.5931852883575184,3.732443553442029) -- (2.5146238146362716,3.716075207864889) -- (2.434585379169892,3.699810646965276) -- (2.3530234098112928,3.6836516590341866) -- (2.2698893735806824,3.667600020754857) -- (2.1851326721181277,3.6516574970074074) -- (2.0987005303771475,3.635825840674801) -- (2.010537878044144,3.62010679245011) -- (1.920587223122699,3.604502080645129) -- (1.8749227332587366,3.5967431375651673) -- (1.8287885170721692,3.5890134210003475) -- (1.7821766454169812,3.581313143423012) -- (1.7350790108346732,3.5736425164962955) -- (1.6874873225049611,3.566001751068302) -- (1.63939310102419,3.558391057166313) -- (1.5907876730041823,3.550810643991009) -- (1.541662165484384,3.543260719910723) -- (1.4920075001499915,3.535741492455714) -- (1.4418143873476459,3.528253168312457) -- (1.3910733198909215,3.5207959533179682) -- (1.3397745666465124,3.5133700524541434) -- (1.28790816589225,3.5059756698421234) -- (1.2354639184374,3.4986130087366836) -- (1.1824313804951045,3.491282271520649) -- (1.1287998562966832,3.483983659699327) -- (1.0745583904366633,3.4767173738949726) -- (1.0196957599371386,3.469483613841272) -- (0.9642004660192791,3.4622825783778515) -- (0.9080607255693479,3.455114465444814) -- (0.8512644622859195,3.447979472077295) -- (0.7937992974943049,3.440877794400051) -- (0.735652540613478,3.4338096276220624) -- (0.6768111792601375,3.4267751660311756) -- (0.6172618689735334,3.419774602988754) -- (0.5569909225442334,3.4128081309243683) -- (0.4959842989286539,3.405875941330507) -- (0.4342275917306033,3.398978224757309) -- (0.3717060172300704,3.392115170807329) -- (0.3084044019380586,3.3852869681303277) -- (0.24430716965579732,3.3784938044180794) -- (0.17939832801483568,3.3717358663992196) -- (0.11366145447375474,3.365013339834107) -- (0.04707968174569244,3.3583264095097225) -- (-0.020364317370572597,3.351675259234587) -- (-0.08868834578459711,3.3450600718337067) -- (-0.15791069856352682,3.3384810291435527) -- (-0.2280501800685898,3.33193831200706) -- (-0.299126121812726,3.325432100268653) -- (-0.37115840107468834,3.318962572769312) -- (-0.4441674603075288,3.312529907341646) -- (-0.5181743273812545,3.3061342808050096) -- (-0.5932006367016245,3.2997758689606456) -- (-0.6692686512499226,3.293454846586846) -- (-0.7464012855906931,3.287171387434155) -- (-0.8246221298974182,3.280925664220585) -- (-0.9039554750492174,3.2747178486268753) -- (-0.9844263388545073,3.2685481112917714) -- (-1.0660604934612703,3.262416621807333) -- (-1.1488844940169187,3.2563235487142745) -- (-1.2329257086447798,3.250269059497328) -- (-1.3182123498083465,3.244253320580647) -- (-1.4047735071387843,3.238276497323223) -- (-1.4926391818059284,3.2323387540143482)(19.03901248497206,3.261455384408123) -- (18.99099114296469,3.2650351454888358) -- (18.943373290111968,3.2686280365100275) -- (18.89615364104563,3.2722340236540406) -- (18.849327002144012,3.2758530729799507) -- (18.802888269549136,3.279485150423888) -- (18.75683242723499,3.2831302217993534) -- (18.71115454512535,3.2867882527975505) -- (18.665849777259943,3.290459208987696) -- (18.57634061032286,3.297839758612752) -- (18.48826777427216,3.3052715928009997) -- (18.401595353036175,3.3127544317480844) -- (18.316288615616156,3.3202879937293464) -- (18.232313967584002,3.327871995110436) -- (18.149638904942822,3.335506150357988) -- (18.06823197021695,3.343190172050373) -- (17.988062710646687,3.350923770888515) -- (17.909101638371027,3.3587066557067913) -- (17.83132019248905,3.3665385334839852) -- (17.75469070289719,3.3744191093543265) -- (17.67918635580635,3.3823480866185873) -- (17.604781160848255,3.390325166755256) -- (17.531449919686708,3.398350049431772) -- (17.459168196053376,3.40642243251584) -- (17.387912287133762,3.414542012086799) -- (17.317659196232878,3.4227084824470673) -- (17.248386606653977,3.430921536133651) -- (17.18007285672832,3.439180863929722) -- (17.112696915937274,3.447486154876256) -- (17.04623836207108,3.4558370962837426) -- (16.98067735937242,3.46423337374396) -- (16.915994637615604,3.472674671141805) -- (16.852171472074815,3.481160670667202) -- (16.789189664337698,3.489691052827066) -- (16.727031523923106,3.4982654964573303) -- (16.665679850663803,3.506883678735037) -- (16.605117917817214,3.5155452751904965) -- (16.54532945586951,3.524249959719496) -- (16.486298636999955,3.5329974045955836) -- (16.428010060174213,3.5417872804824033) -- (16.37044873683743,3.5506192564460983) -- (16.313600077178783,3.559492999967761) -- (16.257449876941095,3.568408176955968) -- (16.201984304750773,3.577364451759341) -- (16.147189889943768,3.5863614871791984) -- (16.093053510865523,3.5953989444822394) -- (16.039562383623107,3.6044764834133063) -- (15.98670405126967,3.613593762208188) -- (15.934466373401587,3.6227504376064923) -- (15.882837516150353,3.631946164864564) -- (15.831805942551608,3.6411805977684706) -- (15.781360403274935,3.650453388647031) -- (15.73148992769885,3.659764188384912) -- (15.682183815315824,3.6691126464357646) -- (15.633431627453449,3.678498410835427) -- (15.58522317929815,3.6879211282151747) -- (15.537548532208582,3.6973804438150233) -- (15.490397986306528,3.706876001497088) -- (15.443762073333776,3.7164074437589876) -- (15.397631549763855,3.725974411747311) -- (15.351997390157893,3.7355765452711167) -- (15.306850780754942,3.7452134828155077) -- (15.21798597900827,3.764590317368343) -- (15.13097063827776,3.784101997325808) -- (15.045741264979885,3.8037455843006693) -- (14.962237209044742,3.8235181200409842) -- (14.88040050603188,3.843416626875598) -- (14.80017572964523,3.8634381081625744) -- (14.721509853856638,3.8835795487404767) -- (14.644352123915292,3.9038379153824407) -- (14.56865393558215,3.9242101572529653) -- (14.494368721983559,3.944693206367359) -- (14.421451847529257,3.965283978053769) -- (14.34986050838519,3.985979371417718) -- (14.279553639033399,4.006776269809088) -- (14.210491824488978,4.02767154129148) -- (14.142637217778397,4.048662039113866) -- (14.07595346231476,4.069744602184484) -- (14.010405618834419,4.090916055546883) -- (13.945960096584828,4.112173210858061) -- (13.882584588478124,4.133512866868619) -- (13.820248009945935,4.154931809904855) -- (13.75892044125134,4.176426814352731) -- (13.698573073031891,4.197994643143646) -- (13.639178154864531,4.219632048241913) -- (13.580708946658323,4.2413357711339135) -- (13.523139672695262,4.263102543318814) -- (13.466445478152323,4.284929086800786) -- (13.410602387949652,4.306812114582668) -- (13.355587267780935,4.328748331160965) -- (13.301377787192108,4.350734433022152) -- (13.247952384583687,4.372767109140156) -- (13.19529023402077,4.394843041474998) -- (13.143371213742611,4.4169589054724705) -- (13.092175876270979,4.439111370564803) -- (13.041685420023338,4.461297100672237) -- (12.991881662343028,4.483512754705422) -- (12.942747013864604,4.505754987068577) -- (12.894264454137758,4.528020448163323) -- (12.846417508438076,4.5503057848931165) -- (12.79919022569789,4.572607641168219) -- (12.752567157494289,4.594922658411106) -- (12.70653333803574,4.617247476062262) -- (12.616175881817385,4.661913063478106) -- (12.528007080597623,4.706577498536044) -- (12.441922748184972,4.75121387705051) -- (12.357824794481825,4.795795311736147) -- (12.275620774809928,4.8402949484037086) -- (12.19522347830903,4.884685982136013) -- (12.116550551500623,4.928941673434229) -- (12.039524153548506,4.973035364324745) -- (11.964070640132897,5.0169404944169385) -- (11.89012027319244,5.06063061690216) -- (11.817606954084813,5.1040794144843025) -- (11.746467977977792,5.147260715232359) -- (11.676643807512365,5.190148508345412) -- (11.608077863983082,5.232716959820566) -- (11.54071633446004,5.274940428014391) -- (11.474507993436669,5.316793479088483) -- (11.409404037728395,5.358250902329864) -- (11.345357933473167,5.399287725336962) -- (11.282325274196584,5.439879229062068) -- (11.220263649003979,5.48000096270115) -- (11.15913252005104,5.51962875842211) -- (11.098893108523901,5.558738745922586) -- (11.039508288431048,5.597307366808526) -- (10.980942487573255,5.635311388784879) -- (10.923161595114944,5.67272791964987) -- (10.866132875232061,5.709534421084401) -- (10.809824886357815,5.745708722228289) -- (10.754207405589353,5.781229033035172) -- (10.699251357856285,5.816073957398007) -- (10.644928749485686,5.850222506037297) -- (10.591212605829181,5.883654109144235) -- (10.538076912645367,5.916348628771199) -- (10.48549656095613,5.9482863709620935) -- (10.43344729511803,5.979448097615249) -- (10.38190566387086,6.009815038071744) -- (10.330848974144113,6.039368900422137) -- (10.280255247419264,6.068091882524845) -- (10.230103178461253,6.095966682729466) -- (10.180372096246735,6.1229765102986615) -- (10.131041926929502,6.1491050955222715) -- (10.082093158695248,6.174336699517561) -- (10.033506808368491,6.19865612370976) -- (9.985264389644277,6.222048718987094) -- (9.937347882826163,6.244500394524874) -- (9.889739705960013,6.265997626273278) -- (9.84242268726077,6.286527465103732) -- (9.795380038735827,6.30607754460899) -- (9.748595330915025,6.324636088552178) -- (9.702052468602863,6.342191917960376) -- (9.655735667573525,6.3587344578584) -- (9.60962943213422,6.374253743638784) -- (9.563718533486288,6.388740427064076) -- (9.517987988817696,6.402185781897873) -- (9.472423041063763,6.414581709161175) -- (9.427009139276329,6.425920742010917) -- (9.336577186410969,6.445401444380068) -- (9.246579131954464,6.460580958287011) -- (9.156904101587232,6.4714227149531744) -- (9.067442792287727,6.477900595637126) -- (8.97808693599347,6.479998994556902) -- (8.8887287739405,6.477712856485704) -- (8.79926053623856,6.47104768893039) -- (8.709573921332009,6.460019548863401) -- (8.61955956999451,6.444655004040106) -- (8.529106528408699,6.424991068994764) -- (8.48368025391044,6.413561232290237) -- (8.438101694688278,6.401075115869254) -- (8.392356282893886,6.387540240901465) -- (8.34642924290335,6.372964760289446) -- (8.300305571636462,6.357357453757707) -- (8.25397001826264,6.340727722564108) -- (8.20740706323323,6.323085583836894) -- (8.160600896577018,6.304441664540751) -- (8.113535395391922,6.2848071950755235) -- (8.066194100462328,6.264194002511437) -- (8.018560191926856,6.242614503464916) -- (7.9706164639169925,6.220081696619273) -- (7.9223452980815665,6.1966091548948015) -- (7.87372863590662,6.172211017272937) -- (7.8247479497338075,6.1469019802794795) -- (7.77538421237376,6.1206972891319475) -- (7.7256178652034,6.093612728556433) -- (7.675428784627912,6.065664613279468) -- (7.624796246779177,6.03686977820065) -- (7.573698890312653,6.007245568251919) -- (7.522114677153905,5.976809827949631) -- (7.470020851034075,5.945580890645677) -- (7.417393893640776,5.913577567484164) -- (7.3642094781964795,5.880819136070281) -- (7.310442420260908,5.8473253288582) -- (7.256066625536636,5.813116321264966) -- (7.201055034438283,5.7782127195176) -- (7.145379563164647,5.742635548240662) -- (7.089011040990232,5.706406237791815) -- (7.031919143467227,5.669546611352972) -- (6.974072321200843,5.632078871784828) -- (6.915437723830144,5.594025588252693) -- (6.85598111881194,5.555409682631661) -- (6.79566680456756,5.5162544156993345) -- (6.734457517510024,5.476583373124406) -- (6.672314332422502,5.4364204512595276) -- (6.60919655560677,5.395789842747039) -- (6.545061610162729,5.354716021946234) -- (6.479864912695446,5.313223730190916) -- (6.41355974067425,5.271337960886139) -- (6.346097089588182,5.229083944453122) -- (6.277425518952053,5.1864871331313855) -- (6.207490986116817,5.143573185647274) -- (6.136236666724937,5.100367951758094) -- (6.063602760524532,5.0568974566811935) -- (5.989526281113337,5.01318788541734) -- (5.9139408280225645,4.969265566977856) -- (5.836776339369226,4.925156958525008) -- (5.757958823100009,4.880888629435204) -- (5.677410064617444,4.836487245294585) -- (5.59504730831515,4.791979551836683) -- (5.510782910249305,4.74739235883179) -- (5.424523958831995,4.702752523937746) -- (5.38061594908684,4.68042126791374) -- (5.35846173537222,4.669254276426959) -- (5.347333859395943,4.663670623749085) -- (5.341757138226985,4.660878781169422) -- (5.340361621926603,4.660180819944783) -- (5.340012659238906,4.660006329630459) -- (5.340001753615834,4.660000876808127) -- (5.3400003904106565,4.660000195205336) -- (5.340000049609279,4.660000024804638) -- (5.340000007009105,4.660000003504551)(13.716485274978753,4.191549404215776);
\draw [line width=0.8pt,dash pattern=on 4pt off 4pt,domain=-1.4286443374600932:19.063698260384477] plot(\x,{(-5.1688-0.*\x)/-1.82});
\draw [fill=qqqqff] (8.98,4.66) circle (2.5pt);
\draw [fill=qqqqff] (8.98,2.84) circle (2.5pt);
\draw[color=qqqqff] (8.49988657575601,2.1619304341304053) node {$0$};
\draw [fill=uuuuuu] (8.98,6.48) circle (1.5pt);
\draw[color=uuuuuu] (9.31667613900048,6.717801109116218) node {$1$};
\draw [fill=xdxdff] (10.738679629525915,4.191549404215776) circle (2.5pt);
\draw [fill=uuuuuu] (13.716485274978753,4.191549404215776) circle (1.5pt);
\end{tikzpicture}
\end{center}
The Witch of Agnesi also enjoys quite a compact formula in polar coordinates:
$$
r(r^2+1)\sin\varphi-r^3\sin^3\varphi=0,
$$
which are similarly dependent on the position of the origin.

But in pedal coordinates, all hell brake loose:
$$
1/36\,{\frac {{C}^{2}}{{p}^{2}}}+1/3\,{\frac {{p}^{2}-2\,{r}^{2}+2}{{p
}^{2}}}+4\,{\frac {{r}^{4}+4\,{p}^{2}-2\,{r}^{2}+1}{{p}^{2}{C}^{2}}}+
192\,{\frac {{r}^{2} \left( {r}^{2}+1 \right) ^{2} \left( p-r \right) 
 \left( p+r \right) }{{p}^{2}{C}^{3}B}}-8\,{\frac { \left( {r}^{2}+1
 \right) ^{2} \left( p-r \right)  \left( p+r \right) }{{p}^{2}CB}}
$$
$$
-
1152\,{\frac { \left( {r}^{2}+1 \right) ^{3} \left( r-1 \right) 
 \left( r+1 \right)  \left( p-r \right)  \left( p+r \right) }{{p}^{2}{
C}^{5}B}}-576\,{\frac {{r}^{2} \left( {r}^{2}+1 \right) ^{4} \left( p-
r \right)  \left( p+r \right) }{{p}^{2}{C}^{4}{B}^{2}}}-82944\,{\frac 
{{r}^{2} \left( {r}^{2}+1 \right) ^{6} \left( p-r \right)  \left( p+r
 \right) }{{p}^{2}{C}^{8}{B}^{2}}}
$$
$$
+13824\,{\frac {{r}^{2} \left( {r}^{
2}+1 \right) ^{5} \left( p-r \right)  \left( p+r \right) }{{p}^{2}{C}^
{6}{B}^{2}}}=0,
$$
where
$$
C:=\sqrt[3]{12 B-108},\qquad B:=\sqrt{-12r^6-36r^4-36r^2+69}.
$$ 
(Again, since this claim is not instrumental to our goal we leave the proof for the interested reader.)
\end{example}
\begin{remark}
From these examples we can see that pedal coordinates are not natural for \textit{all} curves. Similarly horrible pedal equations as for the Witch can be, in fact, expected for \textit{most} algebraic curves, since Cartesian coordinates are indifferent to the position of the origin, which is crucial for the pedal coordinates. 

On the other hand, pedal coordinates are more agreeable with polar coordinates which shares this property. But, obviously, it is required that the polar equation in question does not explicitly mentioned the angle $\varphi$ -- since pedal equation are rotationally invariant.

But instead of $\varphi$ one can use derivatives of $r$:
\end{remark}

\begin{proposition}\label{P2}
A curve $\gamma$ which a solution of a $n$-th order autonomous differential equation ($n\geq 1$)
$$
f\zav{r,r_\varphi^\prime,r_\varphi^{\prime\prime},\dots,r_\varphi^{(n)}}=0,
$$
is the pedal of a curve given in pedal coordinates by
$$
f(p,p_c, p_c p_c^\prime,\dots, (p_c\partial_p)^n p)=0.
$$
In other words
\begin{align*}
r&=P(p), \\
r_\varphi^\prime&=P(p_c),\\
r_\varphi^{\prime\prime}&=P(p_c p_c^\prime), \\
&\vdots \\
r_\varphi^{(n)}&=P\zav{(p_c\partial_p)^{n}p}. 
\end{align*}
\end{proposition}
\begin{remark}
This is a restatement of a similar Proposition in \cite{Blaschke6} -- but with oriented distances $p,p_c$. 
\end{remark}
\begin{proof}
Denoting
$$
r:=\sqrt{x^2+y^2},\qquad \varphi:=\arctan\frac{y}{x},\qquad {\bf x}:=(x,y),\qquad {\bf x}^\perp:=(-y,x),
$$
we have
$$
\dot r=\frac{x\dot x+y\dot y}{r}=\frac{{\bf x}\cdot {\bf \dot x}}{r},\qquad
\dot \varphi=\frac{\dot y x-\dot x y}{r^2}=\frac{{\bf \dot x}\cdot {\bf x}^\perp}{r^2}, 
$$
thus
$$
r'_\varphi:=\frac{\partial r}{\partial \varphi}=\frac{\dot r }{\dot \varphi}=\frac{{\bf x}\cdot {\bf \dot x} r}{{\bf \dot x}\cdot {\bf x}^\perp}=\frac{p_c r}{p}=P(p_c).
$$
From this we can see that for any function $f$ of pedal coordinates it holds
$$
P\zav{p_c \partial_p f}= P(p_c)\partial_{P(p)} P(f)=r'_\varphi \partial_r P(f)=\partial_\varphi P(f).
$$
Hence
$$
P\zav{(p_c \partial_p)^n f}=\partial_\varphi^n P(f).
$$
\end{proof}
\begin{example}\label{Logspiralex}
Logarithmic spiral is given in polar coordinates by
$$
r=r_0e^{\alpha \varphi}.
$$ 
Differentiating with respect to $\varphi$ we get an autonomous differential equation:
$$
r'_\varphi=\alpha r.
$$
Using Proposition \ref{P2} we get an equation which is invariant under Pedal transform:
$$
p_c=\alpha p\qquad \stackrel{P}{\longrightarrow} \qquad p_c=\alpha p.
$$
We can also derive relation:
$$
r^2=(1+\alpha^2) p^2.
$$
\end{example}
But even polar coordinates are \textit{not} most natural starting place from which pedal equations can be generated. That title, we must conclude, belongs to force problems. 

\subsection{Force problems}

Theorem \ref{T1} can be used, with great advantage, to interpret a large family of pedal equations. For instance, there is the Kepler problem:
\begin{equation}\label{Kproblem}
{\bf \ddot x}=-\frac{M}{r^3}{\bf x}, 
\end{equation}
i.e. evolution of a test particle under the influence of a single gravitationally attracting body. It is well known that solutions are focal conic sections. Using Theorem \ref{T1}, these curves are thus given in pedal coordinates as
$$
\frac{L^2}{p^2}=\frac{2M}{r^2}+c,
$$
which explains the equation (\ref{pedellipsefocus}).

Similarly, solutions of a Hook's law 
\begin{equation}\label{Hlaw}
{\bf \ddot x}=-\omega {\bf x},
\end{equation}
i.e. \textit{central} conic sections are hence given in pedal coordinates as
$$
\frac{L^2}{p^2}= -\omega r^2+c,
$$
which explains (\ref{pedellipsecenter}).

It is also possible to write down a dynamical system whose solution would contain ellipses which passes through origin -- just reinterpret equation (\ref{pedellipseboundary}) using Theorem \ref{T1}. But it would be a horrible expression involving both Lorentz-like term and a central term with nested square roots.

As another example, take a problem of determining the orbit of a charged particle in a uniform magnetic field:
$$
{\bf \ddot x}= 2a {\bf \dot x}^\perp,
$$
(i.e. experiencing only Lorentz force). Solution are, of course, circles.  In pedal coordinates we have 
$$
\frac{\zav{L-a r^2}^2}{p^2}=c\qquad \Rightarrow\qquad \sqrt{c}p= a r^2-L.
$$
This is thus another derivation for pedal equation of circles.

More generally, we can consider equation for ``epicycles'':
$$
{\bf \ddot x}= -a {\bf x}+ 2b {\bf \dot x}^\perp,\qquad a>0,b\in\R.
$$
This is easier to solve in complex numbers. Introducing $z:=x+\ii y$, where ${\bf x}=(x,y)$ and $\ii \dot z ={\bf \dot x}^\perp$ we obtain: 
$$
\ddot z=-a z+2b\ii \dot z,
$$
with solution:
$$
z=c_1e^{\ii \omega_+ t}+c_2 e^{\ii \omega_- t},\qquad \omega_{\pm}:=b\pm\sqrt{b^2+a}.
$$
From this form we can clearly see that the resulting curves are indeed ``epicycles'', i.e. we are tracing a point on rotating circle (a.k.a ``epicycle'') with angular velocity $\omega_-$ whose center is also rotating on another circle (a.k.a ``deferent'') with angular velocity $\omega_+$.
  
Using Theorem \ref{T1} on the original equation we thus obtain pedal equation for epicycloids:
\begin{equation}\label{epi}
\frac{\zav{L-r^2b}^2}{p^2}=-ar^2+c.
\end{equation}

And so on.
\subsection{Transforms} Even if a pedal equation for a curve is not know or cannot be interpreted via Theorem~\ref{T1}, it is often possible to connect them with other curves using one of the many \textit{transformations of curves} -- some of which are particularly easy to handle in pedal coordinates.

We will list few of them:

\subsubsection{Scaling} From the definition of quantities $p,p_c,r$ it is clear that a classical radial scaling $S_\alpha$, i.e. transform that maps a given point ${\bf x}$ to the point
$$
{\bf \tilde x}=\frac{{\bf x}}{\alpha},
$$
is given in pedal coordinates as follows:
$$
f(p,r,p_c)=0 \qquad \stackrel{S_\alpha}{\longrightarrow}\qquad f(\alpha p, \alpha r,\alpha p_c)=0.
$$
\subsubsection{Pedal} As was mentioned, pedal transform $P$ maps any point on a given curve ${\bf x}$ to the orthogonal projection ${\bf \tilde x}$ of origin to the tangent at ${\bf x}$. 

Algebraically, the new point is given by 
$$
{\bf \tilde x}={\bf x}-p_c \frac{\bf \dot x}{\abs{\bf \dot x}}. 
$$

This transformation has very nice properties. It maps a focal ellipse or hyperbola (that is with the origin at focus) to its circumcircle. Pedal of a parabola is a line  -- not the directrix, though; this line is parallel to the directrix and passes through the vertex. Central rectangular hyperbola is mapped to Lemniscate of Bernoulli and so on. 

Using the already mentioned formula:
$$
f(p,r, p_c)=0 \qquad \stackrel{P}{\longrightarrow}\qquad f\zav{r,\frac{r^2}{p},\frac{r}{p}p_c}=0.
$$
we can check these claims with ease. For instance, the following line of computation: 
$$
\frac{L^2}{p^2}=\frac{M}{r}+c \qquad \stackrel{P}{\longrightarrow} \qquad \frac{L^2}{r^2}=\frac{M p}{r^2}+c \qquad \Rightarrow \qquad Mp= -c r^2+L^2.
$$
proves that the pedal of a (focal) conic is, indeed, a circle. Or, inversely, taking for granted that pedal of a central rectangular hyperbola ($\frac{a^2 b^2}{p^2}=r^2$) is Lemniscate of Bernoulli, we can in no time produce pedal equation for it!
$$
\frac{a^2 b^2}{p^2}=r^2 \qquad \stackrel{P}{\longrightarrow} \qquad \frac{a^2 b^2}{r^2}=\frac{r^4}{p^2} \qquad \Rightarrow \qquad r^3=ab p.
$$

\subsubsection{Circle inverse} The circle inverse $I$, i.e. transform
$$
{\bf \tilde x}=\frac{{\bf x}}{r^2},
$$
is given in pedal coordinates as follows:
$$
f(p,r,p_c)=0, \qquad \stackrel{I}{\longrightarrow} \qquad f\zav{\frac{p}{r^2},\frac{1}{r},-\frac{p_c}{r^2}}=0.
$$
This can be seen performing following computations:
\begin{align*}
\tilde r^2&:=\abs{{\bf \tilde x}}^2=\abs{\frac{{\bf x}}{r^2}}^2=\frac{1}{r^2}.\\
{\bf \dot{\tilde x}}&={\bf \dot  x}\frac{1}{r^2}-2{\bf x} \frac{p_c\abs{{\bf \dot x}}}{r^4}\\
\abs{{\bf \dot{\tilde x}}}&=\frac{\abs{{\bf \dot x}}}{r^2},\\
\tilde p&:=\frac{{\bf \tilde x}^\perp \cdot {\bf \dot {\tilde x} } }{{\bf \dot{\tilde x}}}= \frac{{\bf x}^\perp}{r^2}\cdot \zav{{\bf \dot x}\frac{1}{r^2}-2{\bf x} \frac{p_c\abs{{\bf \dot x}}}{r^4}} \frac{r^2}{{\abs{\bf \dot x}}}=\frac{{\bf x}^\perp \cdot {\bf \dot x}}{\abs{{\bf \dot x}} r^2}=\frac{p}{r^2}\\
\tilde p_c&:=\frac{{\bf \tilde x} \cdot {\bf \dot {\tilde x} } }{{\bf \dot{\tilde x}}}= \frac{{\bf x}}{r^2}\cdot \zav{{\bf \dot  x}\frac{1}{r^2}-2{\bf x} \frac{p_c\abs{{\bf \dot x}}}{r^4}} \frac{r^2}{{\abs{\bf \dot x}}}=-\frac{p_c}{r^2}.
\end{align*}
It is an easy exercise to verify that circle inverse possess its very known properties, e.g. that maps circles and lines to circles and lines.

\subsubsection{Dual curve}
Dual curve $D$, i.e. a curve in the dual projective space consisting of the set of lines tangent to the original curve. It is given by 
$$
{\bf \tilde x }:=\frac{{\bf \dot x}}{{\bf x}^\perp \cdot {\bf \dot x}}.
$$
Its action in pedal coordinates can be establish performing similar (though more difficult) computation as in the case of circle inverse. But a shortcut can be made since it is known that Dual can be obtained by composing Pedal transform with Circle inverse as follows: $D=IP$.

Thus
$$
f(p,r,p_c)=0, \qquad \stackrel{D:=IP}{\longrightarrow} \qquad f\zav{\frac{1}{r},\frac{1}{p},-\frac{p_c}{r p}}=0.
$$

We can easily verify, that Dual is its own inverse, i.e. $D^{-1}=D$ as is the case also with $I^{-1}=I$.

For central force problems where angular momentum is conserved (that is the quantity ${\bf x}^\perp {\bf \dot x}$ is constant) we can use the Dual transform to obtain so-called ``curve of velocities'' of the orbit -- i.e. curve witch is traced by (suitably scaled) vector ${\dot x}$ positioned at the origin. Newton famously showed in his \textit{Principia} \cite{Newton} by geometrical arguments that the curve of velocities for the Kepler problem is a circle.

In pedal coordinates obtaining the same results amounts to show that the Dual of focal conic is a circle:
$$
\frac{L^2}{p^2}=\frac{M}{r}+c\qquad \stackrel{D}{\longrightarrow} \qquad L^2 r^2 =M p+c.
$$

Dual transform can be also used to define new transforms -- their dual counterparts. For every transform $T$ we can define its \textit{dual transform} $T^*$ as follows:
$$
T^*:=D T D.
$$

Dual transform to Pedal is its inverse: $P^*=P^{-1}$ and dual Inverse is given by
$$
f(p,r,p_c)=0, \qquad \stackrel{I^*:=DID}{\longrightarrow} \qquad f\zav{\frac{1}{p},\frac{r}{p^2},-\frac{p_c}{p^2}}=0.
$$

Hence dual Inverse of a line is again a line. Focal conics are also preserved under this transform.
But dual Inverse of a circle is quite complicated. Even for the circle passing through origin we get a curve known as ``Tschirnhausen cubic''.

\subsubsection{Parallel curves}
Another important transform for our purposes will be \textit{parallel curves} $E_\alpha$, i.e. curves that are of equal distance $\alpha$ to the given curve. 
$$
{\bf \tilde x}={\bf x}+\alpha \frac{{\bf \dot x}^\perp}{\abs{\bf \dot x}}.
$$
Deriving its effect is easy in pedal coordinates, since it is known that parallel curves shares normals and thus the variable $p_c$ is preserved. The distance to tangent $p$ is hence just shifted, i.e. $p\to p-\alpha$ and  from the identity $r^2=p^2+p_c^2$ we can work out the remaining variable. Altogether we get
$$
f\zav{p,r^2,p_c}=0 \qquad \stackrel{E_\alpha}{\longrightarrow}\qquad f\zav{p-\alpha,r^2-2p\alpha+\alpha^2,p_c}=0.
$$

The dual version of parallel transform is also very interesting. It can be readily computed in pedal coordinates to have the following action:
$$
f\zav{\frac{1}{p^2},r,\frac{p}{p_c}}=0 \qquad\stackrel{E^*_\alpha:=D E_\alpha D}{\longrightarrow}\qquad  f\zav{\frac{1}{p^2}-\frac{2\alpha}{r}+\alpha^2,\frac{r}{1-\alpha r},\frac{p}{p_c}(1-\alpha r)}=0.
$$
In Cartesian coordinates the same think is obtained when for any given point ${\bf x}$ on a curve the transformed point ${\bf \tilde x}$ is given by
$$
{\bf \tilde x}:=\frac{{\bf x}}{1+\alpha r}.
$$

It is easy to see that the dual parallel of a line is a focal conic:
$$
\frac{1}{p^2}=c \qquad\stackrel{E^*_\alpha}{\longrightarrow}\qquad  \frac{1}{p^2}=\frac{2\alpha}{r}+c-\alpha^2,
$$
and that focal conics are preserved by this operation since it holds:
$$
E^*_\alpha E^*_\beta=E^*_{\alpha+\beta},
$$
i.e. $E^*_\alpha$ is a group of transformations (as is regular $E_\alpha$).
 
\subsubsection{Curves Harmonics}
The $k$-th harmonic of a curve is done simply by multiplying the angle of every point on a curve by a constant $k$. For example, the figure below shows the second harmonic ($k=2$) and the third subharmonic ($k=\frac{1}{3}$) of an ellipse, where the center of polar coordinates are in one focus.
\begin{center}
\includegraphics{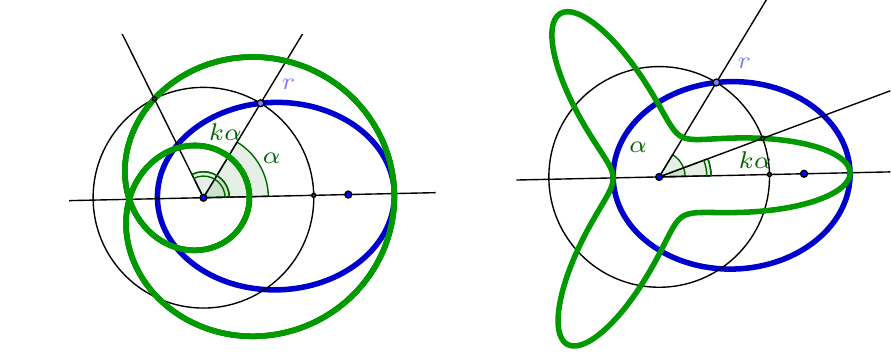}
\\
The second harmonic and the third subharmonic of a focal ellipse.
\end{center}

In pedal coordinates this amounts to the following:
$$
f\zav{\frac{1}{p^2},r}=0\qquad \stackrel{H_k}{\longrightarrow} \qquad f\zav{\frac{k^2}{p^2}-\frac{k^2-1}{r^2},r}=0.
$$

Many interesting curves can be obtained by making harmonics of known curves. E.g.
harmonics of a line are so-called Cote's spirals (or epispirals) and harmonics of a circle are called ``roses''.

But most famous application of harmonic curves gave Newton himself in \textit{Principia}. He proved his ``Theorem of revolving orbits'' (see \cite{Newton}), that if a curve is given as a solution to the central force problem $F(r)$ adding additional force of the form
$
\frac{L^2}{m r^3}(1-k^2),
$
where $L$ is the particle's angular momentum and $m$ its mass, is equivalent to a making the curve's $k$-th harmonic.

(This particular fact is easy to see in pedal coordinates and it follows directly by applying $H_k$ transform on central force case of Theorem \ref{T1}.)

Theorem of revolving orbits remained largely forgotten until 1997, when it was studied in works \cite{Bell1},\cite{Bell2}. A generalization was discovered by Mahomed and Vawda in 2000 \cite{Mahomed}. They assumed that the radial distance $r$, and the angle $\varphi$ changes according to rule:
\begin{equation}\label{Mahomedtr}
r\to \frac{ar}{1-br},\qquad \varphi \to \frac{1}{k}\varphi,
\end{equation}
where $a,b$ are given constant.

They proved that such a transform of the solution is equivalent to changing the force as follows:
\begin{equation}\label{Mohamedeq}
F(r)\to \frac{a^3}{(1-br)^2}F\zav{\frac{ar}{1-br}}+\frac{L^2}{m r^3}(1-k^2)-\frac{bL^2}{mr^2},
\end{equation}
where again $m$ is the particle mass and $L$ its angular momentum.

As we can see, this transform \ref{Mahomedtr} is nothing more than a combination of Harmonic curve $H_k$ an Dual parallel transform $E_\alpha$ and scaling.

Further generalization to cover also non-local transformation was done in \cite{Blaschke6}.
\subsubsection{Dual Harmonics}
Transform that will be useful in what follows is also dual version of $H_k$, which can easily be shown to be given by
$$
f\zav{p,p_c}=0\qquad \stackrel{H_k^\star:=D H_k D}{\longrightarrow} \qquad f\zav{p,kp_c}=0.
$$
As we can see, Dual Harmonics preserves lines and maps Involute of a circle to a different Involute o a circle.

Interestingly, for a special value of $k$, it can be shown, that dual harmonics $H^\star_k$ transform a circle that does not contain the origin to a curve with \textit{linear} pedal equation. More precisely it holds: For $a>R$:
$$
2Rp=r^2+R^2-a^2, \qquad \stackrel{H^\star_k}{\longrightarrow} \qquad Rp=ar+R^2-a^2,\qquad k=\frac{a}{\sqrt{a^2-R^2}}.
$$  
\subsubsection{Power transform}
Dealing with planar curves only allow us to describe a point on a curve ${\bf x}$ not as a vector, but as a complex number:
$$
{\bf x}=(x,y),\qquad z:= x+\ii y.
$$
With this notation we can condense both pedal coordinates $p,p_c$ into a single quantity:
$$
\frac{z \dot {\bar z}}{\abs{\dot z}}=p_c+\ii p.
$$

We are now ready to introduce very useful \textit{power transform} $M_\alpha$ which maps a point $z$ to its power $z^\alpha$:
$$
M_\alpha: \qquad \tilde z:=z^\alpha,\qquad \alpha\in\mathbb{R}.
$$
This transform translates into pedal coordinates very easily:
\begin{align*}
\tilde r&:=\abs{\tilde z}=\abs{z^\alpha}=r^\alpha.\\
\dot{\tilde z}&=\alpha z^{\alpha-1}\dot z.\\
\abs{\dot{\tilde z}}&= |\alpha| r^{\alpha-1} \abs{\dot z}.\\
\tilde p_c+\ii\tilde p&:=\frac{\tilde z\overline{\dot{\tilde z}}}{\abs{\dot{\tilde z}}}=\frac{z^\alpha \alpha \bar z^{\alpha-1}\dot z}{|\alpha|r^{\alpha-1}\abs{\dot z}}=\sgn(\alpha) r^{\alpha-1}\frac{z\dot{\bar z}}{\abs{\dot z}}=\sgn(\alpha) r^{\alpha-1}(p_c+\ii p).
\end{align*}
This yields:
\begin{align*}
f(p,r,p_c)&=0, & &\stackrel{M_\alpha }{\longrightarrow} & f\zav{p r^{\alpha-1},r^\alpha,p_c r^{\alpha-1}}&=0, &\alpha>0.\\
f(p,r,p_c)&=0, & &\stackrel{M_{-\alpha} }{\longrightarrow} & f\zav{-p r^{\alpha-1},r^\alpha,-p_c r^{\alpha-1}}&=0, &\alpha>0.\\
\end{align*}

With this transform \textit{a lot} of new curves can be obtained. For instance, powers of a circle passing through origin 
$$
2Rp=r^2 \qquad \stackrel{M_\alpha }{\longrightarrow}\qquad 2R p=r^{\alpha+1},
$$
are nothing else than famous \textit{sinusoidal spirals}, i.e. family of curves $\sigma_n(a)$, given in polar coordinates
$$
r^n=a^n\sin\zav{n\varphi+\varphi_0}.
$$
Specific examples contains many famous curves, e.g.
\begin{align*}
n& & a^n p &=r^{n+1} & &\text{Curve} &\text{Pedal point:}\\
n&=0 & p&=r & &\text{Concentric circle }\abs{x}=R. & \text{Center.}\\
n&=-1 & p&=a & &\text{Line.} & \text{A point distant }a.\\
n&=1 & a p&=r^2 & &\text{Circle. } & \text{On the circle.}\\
n&=2 & a^2 p &=r^3 & &\text{Lemniscate of Bernoulli.} & \text{Center.}\\
n&=-2 & rp &=a^2 & &\text{Rectangular hyperbola.} & \text{Center.}\\
n&=-\frac12 & a^{-\frac12} p &=r^{\frac12} & &\text{Parabola.} & \text{Focus}\\
n&=\frac12 & a^{\frac12} p &=r^{\frac32} & &\text{Cardioid.} & \text{Cusp.} \\
n&=-\frac13 &  p^3 &= a r^{2} & &\text{Tschirnhausen cubic.} & \text{Center.} \\
\end{align*}
Sinusoidal spirals are famously invariant under a number of transforms, for example: 
\begin{align*}
\sigma_n(a) &\stackrel{P}{\longrightarrow} \sigma_{\frac{n}{n+1}}(a) & \text{Pedal.}\\
\sigma_n(a) &\stackrel{I}{\longrightarrow} \sigma_{-n}\zav{\frac{1}{a}} & \text{Inverse.}\\
\sigma_n(a) &\stackrel{D}{\longrightarrow} \sigma_{-\frac{n}{n+1}}\zav{\frac{1}{a}} & \text{Dual.}\\
\sigma_n(a) &\stackrel{M_\alpha}{\longrightarrow} \sigma_{\alpha n}\zav{a^{\frac{1}{\alpha}}} & \text{Complex power.}
\end{align*}

Let us also mention that the transform $M_{\frac12}$, i.e. complex square root, has very nice properties.

It maps central conics into focal one:
$$
\frac{L^2}{p^2}= a r^2+c \qquad \stackrel{M_{\frac12} }{\longrightarrow} \qquad \frac{L^2}{p^2}= a+\frac{c}{r}.
$$

It also maps a circle into a central Cassini oval, which is the locus of points such that product of distances from two foci is constant:
$$
\abs{z-a^2}=R, \qquad \stackrel{M_{\frac12} }{\longrightarrow} \qquad \abs{\tilde z^2-a^2}=R,\qquad \Rightarrow \qquad \abs{\tilde z-a}\abs{\tilde z+a}=R.
$$
This gives us rather nice pedal equation for a central Cassini oval:
$$
2pR= r^2+R^2-\abs{a}^2\qquad \stackrel{M_{\frac12} }{\longrightarrow} \qquad
2p R= (r+ R^2-\abs{a}^2)\sqrt{r}.
$$
\begin{example} (Shifted sinusoidal spirals)
As we saw, sinusoidal spirals are just rescaled complex powers of a circle passing through the origin, i.e.
$$
\sigma_n(a):=S_{\beta}M_\alpha(ap=r^2).
$$ 
This begs the question what are rescaled powers of a circle in general position, that is curves of the form
$$
S_{\beta}M_\alpha(2Rp=r^2+R^2-a^2).
$$
$$
2Rp=r^2+R^2-a^2\qquad \stackrel{M_\alpha}{\longrightarrow}\qquad 
 2Rp r^{\alpha-1}=r^{2\alpha}+R^2-a^2\qquad \stackrel{S_\beta}{\longrightarrow}\qquad 2Rp\beta= \beta^{2\alpha} r^{1+\alpha}+(R^2-a^2) r^{1-\alpha}.
$$
This gives us with little bit of cleaning the following family of curves:
$$
\sigma_\alpha(a,b): \qquad p=a r^{\alpha+1}+b r^{1-\alpha},
$$
which we will call \textit{Shifted sinusoidal spirals}.
Obviously, the original sinusoidal spirals are a special case: 
$$
\sigma_n(a)=\sigma_n(a^{-n},0).
$$
Also, here is a short list of properties
\begin{align*}
\sigma_1(a,b) &\longrightarrow  2Rp=r^2+R^2-\abs{a}^2, & R=\frac{1}{2a},\quad \abs{a}=\frac{\sqrt{1-4ab}}{2a}.\\
\sigma_\alpha(a,b)&=\sigma_{-\alpha}(b,a) &\text{Symmetry}\\
\sigma_\alpha(a,b) &\stackrel{S_\beta}{\longrightarrow} \sigma_{\alpha}\zav{a\beta^\alpha,b\beta^{-\alpha}} & \text{Scaling}\\
\sigma_\alpha(a,b) &\stackrel{I}{\longrightarrow} \sigma_{-\alpha}\zav{a,b} & \text{Inverse}\\
\sigma_\alpha(a,b) &\stackrel{M_\beta}{\longrightarrow} \sigma_{\alpha \beta}\zav{a,b} & \text{Complex power}.
\end{align*}
As we can see we have lost the invariance under Pedal transform.

Classical sinusoidal spirals can be interpreted via Theorem \ref{T1} as the zero energy solution to a central force problem with force varying as a power of the distance, i.e.:
$$
{\bf \ddot x}=M r^{\alpha-1} {\bf x},\qquad \stackrel{Th. \ref{T1}}{\Rightarrow} \qquad\frac{L^2}{p^2}=\frac{2M}{\alpha+1} r^{\alpha+1}+c. 
$$
We get $n=\alpha$ sinusoidal spiral when the quantity $c:=\abs{\bf\dot x}^2-\frac{2M}{\alpha+1}r^{\alpha+1}$ equals zero.

We can also, with the aid of Theorem \ref{T2}, reinterpret shifted sinusoidal spiral as a zero energy solution of a force problem, namely:
$$
{\bf \ddot x}=-\zav{a r^{\alpha-1}+ b r^{-1-\alpha}}\abs{\bf \dot x}^3 {\bf x},\qquad \stackrel{Th. \ref{T2}}{\Rightarrow} \qquad p=\frac{La}{\alpha+1} r^{\alpha+1}+\frac{Lb}{1-\alpha} r^{1-\alpha}+c. 
$$
Again, shifted sinusoidal spirals $\sigma_{\alpha}$ are solutions when
$$
c:=\frac{1}{\abs{\bf \dot x}}-\frac{a}{\alpha+1}r^{\alpha+1}-\frac{b}{1-\alpha}r^{1-\alpha}=0.
$$

\end{example}

\begin{remark}
Much more detailed exposition of pedal coordinates, greater list of curves and transforms given in pedal coordinates can be found in \cite{Blaschke6,Yates,Zwikker}. 
\end{remark}

%
\section{Proofs}
\subsection{Proof of Theorem \ref{T2}}
\begin{proof}
We are going to demonstrate the technique on the formula (\ref{T2F2}). We want to wish to express the solution of 
\begin{equation}\label{tosolve}
\ddot {\bf x}=f  \frac{{\bf x}}{\abs{\dot {\bf x}}^\alpha}+g \frac{\dot {\bf x}^\perp}{\abs{\dot {\bf x}}^\beta}, 
\end{equation}
where $f,g$ are functions of pedal coordinates and $\alpha\not=-2$, $\beta\not=-1$.

Multiply (\ref{tosolve}) by $ {\bf \dot x}$ to obtain:
\begin{align*}
 {\bf \ddot x}\cdot {\bf \dot x}&= f  \frac{{\bf x}\cdot {\bf \dot x}}{\abs{ {\bf \dot x}}^\alpha}+g \frac{ {\bf \dot x}^\perp\cdot {\bf \dot x}}{\abs{ {\bf \dot x}}^\beta}\\
\abs{{\bf \dot x}} \partial_t \abs{{\bf \dot x}}&= f \frac{r \dot r}{\abs{ {\bf \dot x}}^\alpha}\\
\abs{{\bf \dot x}}^{\alpha+1}\partial_t \abs{{\bf \dot x}}&=f r\dot r\\
\abs{{\bf \dot x}}^{\alpha+2}&=(\alpha+2)\int f r{\rm d}r.
\end{align*}
Similarly multiply (\ref{tosolve}) by $ {\bf x}^\perp$ to obtain:
\begin{align*}
 {\bf \ddot x}\cdot {\bf x}^\perp&= f  \frac{{\bf x}\cdot {\bf x}^\perp}{\abs{ {\bf \dot x}}^\alpha}+g \frac{ {\bf \dot x}^\perp\cdot {\bf x}^\perp}{\abs{ {\bf \dot x}}^\beta}\\
\partial_t \zav{{\bf \dot x}\cdot {\bf x}^\perp}&= g \frac{ {\bf \dot x}\cdot {\bf x}}{\abs{ {\bf \dot x}}^\beta}\\
\zav{{\bf \dot x}\cdot {\bf x}^\perp}^\beta\partial_t \zav{{\bf \dot x}\cdot {\bf x}^\perp}&= g \frac{ {\bf \dot x}\cdot {\bf x}}{\abs{ {\bf \dot x}}^\beta}\zav{{\bf \dot x}\cdot {\bf x}^\perp}^\beta\\
\partial_t\frac{\zav{{\bf \dot x}\cdot {\bf x}^\perp}^{\beta+1}}{\beta+1}&= g  r\dot r p^\beta\\
\zav{{\bf \dot x}\cdot {\bf x}^\perp}^{\beta+1}&=(\beta+1)\int g p^\beta r{\rm d}r\\
p^{\beta+1}\abs{{\bf \dot x}}^{\beta+1}&=(\beta+1)\int g p^\beta r{\rm d}r\\
\abs{{\bf \dot x}}^{\beta+1}&=\frac{(\beta+1)\int g p^\beta r{\rm d}r}{p^{\beta+1}}\\
\abs{{\bf \dot x}}^{(\beta+1)(\alpha+2)}&=\zav{\frac{(\beta+1)\int g p^\beta r{\rm d}r}{p^{\beta+1}}}^{\alpha+2}\\
\zav{(\alpha+2)\int f r{\rm d}r}^{\beta+1}&=\zav{\frac{(\beta+1)\int g p^\beta r{\rm d}r}{p^{\beta+1}}}^{\alpha+2},
\end{align*}
which is what we want.

The formula (\ref{T2F1}) is just a special case for $\alpha=\beta=0$. The formulas (\ref{T2F3})-(\ref{T2F6}) are limiting cases. The formula (\ref{T2F7}) is actually a corollary of (\ref{T2F1}) with $g=0$. Since we have no Lorentz-like term, we have a conservation law:
$$
L={\bf \dot x}\cdot {\bf x}^\perp.
$$ 
But then we can replace any occurrence of the quantity $\abs{\bf \dot x}$ with 
$$
\abs{\bf \dot x}=\frac{L}{p},
$$
which makes our central term $f$ dependent on pedal coordinates only and thus suitable for (\ref{T2F1}). 
\end{proof}

\subsection{Proof of Proposition \ref{P1}}
\begin{proof}
Since the functional
$$
\mathcal{L}=\inte{s_0}{s_1}f(r)\dd s=\inte{\varphi_0}{\varphi_1}f(r)\sqrt{{r'_\varphi}^2+r^2}\dd \varphi,
$$
does not depend explicitly on $\varphi$, we can use Beltrami identity to obtain Euler-Lagrange equation in the form
$$
f(r)\sqrt{{r'_\varphi}^2+r^2}-f(r)\frac{{r'_\varphi}^2}{\sqrt{{r'_\varphi}^2+r^2}}=L,
$$
substituting $r'_\varphi=\frac{rp_c}{p}$ (Proposition \ref{P2}) we get what we want.
\end{proof}
\section{Variational problems}
\begin{example}(Heron's problem)
Heron of Alexandria argue that the light travel so that is always picks the shortest path. It was obvious to him that shortest path between two points is a straight line. But we can actually prove that finding a minimum of the functional that evaluates arc length which in polar coordinates has the form:
$$
\mathcal{L}=\inte{\varphi_0}{\varphi_1}\sqrt{{r'_\varphi}^2+r^2}\dd \varphi.
$$
Hence $f\equiv 1$ and thus the pedal equation becomes
$$
\frac{L}{p}=1,\qquad p=L,
$$
which indeed belongs to a line.
\end{example} 
\begin{example}\label{CBex}(Central Brachistochrone)
The classical problem of Brachistochrone involves finding a path between two points $A,B$ on which a bead slides frictionlessly under influence of gravity in shortest time possible.

The problem is solved, realizing that minimizing  the time corresponds to the maximizing the speed of travel since
$$
\dd t=\frac{\dd s}{\abs{\bf \dot x}},
$$
that is (instantaneous)time is (instantaneous) path divided by (instantaneous) speed.

Since total energy
$$
E=\frac12 \abs{\bf \dot x}^2+U,
$$
is conserved, we can relate speed to the potential energy as 
$$
\abs{\bf \dot x}=\sqrt{2E-2U}.
$$
Substituting this we get functional to minimize
$$
\mathcal{L}:=\inte{t_0}{t_1}\dd t=\inte{s_0}{s_1}\frac{1}{\sqrt{2E-2U}}\dd s. 
$$

Potential energy is normally given to rise in proportion with elevation $y$, $U=-gy$, but this is only a limiting approximation to the ``correct'' Newtonian potential $U=-\frac{M}{r}$ induced by a central body.
 
If we assume this correct $U$, we end up with the functional for the central Brachistochrone:
$$
\mathcal{L}=\inte{s_0}{s_1}\frac{1}{\sqrt{\frac{2M}{r}-\frac{2M}{r_0}}}\dd s,
$$
where $r_0$ is the starting distance. This is a variational problem
on which Proposition \ref{P1} can be applied to get
$$
\frac{L}{p}=\frac{1}{\sqrt{\frac{2M}{r}-\frac{2M}{r_0}}}.
$$
This is an interesting curve. It can be also written in the form
$$
\zav{\frac{L}{p}}^{-2}=\frac{2M}{r}-\frac{2M}{r_0},
$$
which can be readily interpreted by Theorem \ref{T2} as a solution to
$$
{\bf \ddot x}=\frac{M\abs{\dot x}^4}{2r^3}{\bf x}.
$$

\end{example}
\begin{example}(Gravity train)
Gravity train was first described in a note from Robert Hook to Isaac Newton. It is a problem of finding the shape of a tunnel through Earth connecting points $A,B$ on its surface that a train would traverse using only gravity (and again no friction is assumed) in shortest possible time. 

It is again a Brachistochrone problem but with different potential energy $U$. Assuming that Earth is a perfect ball of homogeneous density we can use Newton's Shell theorem to obtain potential energy of the form
$$
U:=\alpha r^2,
$$
for some constant $\alpha$ (depending on the density). We will choose units so that $\alpha=1$.

Again, from the conservation of total energy we have
$$
E=\frac12\abs{\bf \dot x}^2+U,\qquad \Rightarrow \qquad \abs{\bf \dot x}=2\sqrt{E-r^2}.
$$ 
Our train is not moving in departure and destination stations $A,B$ located on the surface of the Earth. Thus we must have $E=R^2$, where $R$ is the Earth's radius.

By exactly the same logic as in the previous example, we thus obtain a functional to minimize with 
$$
f\propto\frac{1}{\sqrt{R^2-r^2}}.
$$ 

Hence we obtain solutions in the form
$$
\frac{L}{p}=\frac{1}{\sqrt{R^2-r^2}},\qquad \Longrightarrow\qquad  \frac{L^2(R^2-r^2)^2}{p^2}=R^2-r^2.
$$ 
Comparing with (\ref{epi}), we can recognize this equation to belong to particular species of epicycles. In fact, further analysis reveals that these are exactly ``hypocycloids'', i.e.  a circle revolving inside a bigger circle. This is surprisingly nice result remembering that the solution to classical Brachistochrone is the cycloid, i.e. a circle revolving on a line. 
\begin{center}
\definecolor{qqccqq}{rgb}{0.,0.8,0.}
\definecolor{qqttcc}{rgb}{0.,0.2,0.8}
\definecolor{qqqqff}{rgb}{0.,0.,1.}
\definecolor{uuuuuu}{rgb}{0.26666666666666666,0.26666666666666666,0.26666666666666666}
\begin{tikzpicture}[line cap=round,line join=round,>=triangle 45,x=1.0cm,y=1.0cm,scale=0.5]
\clip(4.22,2.32) rectangle (15.34,11.6);
\draw [line width=1.2pt,color=qqttcc] (9.68,6.94) circle (3.70264770130781cm);
\draw [color=qqccqq] (11.841901199278373,8.815767217020941) circle (0.8404248854020063cm);
\draw (9.68,6.94)-- (13.38,6.8);
\draw[line width=1.6pt] (11.121535660791613,8.38288446272868)(6.70909968517315,7.77336855083781) -- (6.7090997135550605,7.773368525916794) -- (6.709099770318894,7.773368476074762) -- (6.709099883846552,7.773368376390685) -- (6.709100110901888,7.773368177022485) -- (6.709100565012655,7.7733677782859) -- (6.709101473234482,7.7733669808119865) -- (6.709103289679407,7.773365385861188) -- (6.709106922574341,7.773362195947717) -- (6.7091141883844525,7.7733558160732645) -- (6.7091287200857295,7.773343056134316) -- (6.709157783812477,7.773317535496256) -- (6.70921591256265,7.773266491179467) -- (6.709332175248604,7.773164390383177) -- (6.709564721355129,7.772960140139409) -- (6.71002989644319,7.772551445044397) -- (6.710960577612616,7.771733276402969) -- (6.712823259869593,7.770093825146812) -- (6.716553871604138,7.766802465467565) -- (6.724035824180233,7.760169908338276) -- (6.7390805650876935,7.746705387571737) -- (6.769476708906255,7.718978555388705) -- (6.831361749930493,7.660336881300446) -- (6.894448024609497,7.597463040960349) -- (6.958428512328412,7.530398347531965) -- (7.02299748616198,7.459205852044757) -- (7.0878520627537505,7.383970144559449) -- (7.1526937351258395,7.304797045289891) -- (7.217229880608357,7.2218131872026134) -- (7.281175236205948,7.135165492168631) -- (7.34425333388577,7.045020543285045) -- (7.406197888477037,6.951563856513457) -- (7.466754131114659,6.854999055296011) -- (7.525680081438582,6.755546952304538) -- (7.554459872880865,6.704811620822003) -- (7.582747752073719,6.653444542952192) -- (7.6105175925098285,6.601477692888439) -- (7.637744279261621,6.548943915747234) -- (7.664403739874601,6.495876889360314) -- (7.690472973892973,6.442311084992691) -- (7.715930080979902,6.388281727036585) -- (7.740754287596781,6.333824751732035) -- (7.764925972207581,6.278976764966578) -- (7.788426688976284,6.223774999207308) -- (7.8112391899273605,6.1682572696198665) -- (7.8333474455410945,6.112461929430057) -- (7.854736663757801,6.056427824584262) -- (7.875393307366683,6.000194247766586) -- (7.895305109757467,5.943800891830442) -- (7.9144610890147895,5.887287802703996) -- (7.932851560337604,5.830695331828932) -- (7.95046814676791,5.7740640881927625) -- (7.967303788215297,5.717434890015648) -- (7.983352748766018,5.660848716152769) -- (7.998610622267396,5.604346657273903) -- (8.013074336180708,5.547969866882243) -- (8.026742153697743,5.491759512234119) -- (8.039613674118618,5.435756725222367) -- (8.051689831490519,5.380002553285286) -- (8.062972891509284,5.32453791040328) -- (8.073466446688103,5.2694035282457925) -- (8.083175409799557,5.214639907529935) -- (8.092106005599693,5.16028726965283) -- (8.100265760844863,5.106385508658825) -- (8.107663492614268,5.0529741436024596) -- (8.11430929495337,5.000092271367709) -- (8.12021452385538,4.947778520003123) -- (8.125391780600266,4.896071002632183) -- (8.12985489347274,4.845007271997226) -- (8.133618897882759,4.794624275694527) -- (8.139115628329085,4.696044987442614) -- (8.1420255442042,4.600614963610524) -- (8.142510000968327,4.508603618869773) -- (8.140747237039202,4.420266678600519) -- (8.136931385482315,4.33584496881837) -- (8.131271402703621,4.255563280385796) -- (8.123989919642916,4.179629313258819) -- (8.115322021362594,4.108232706112678) -- (8.105513961291738,4.041544156255156) -- (8.094821816720179,3.979714634277113) -- (8.083510092437232,3.9228746974072553) -- (8.060119377743185,3.8245803393509625) -- (8.033977352531334,3.7362270845672305) -- (8.012624634479879,3.676984924159583) -- (7.999660930168899,3.6465612349609144) -- (7.998543828014876,3.6440096104607065) -- (8.012484932018433,3.6677605729191542) -- (8.044352651093881,3.715669287287571) -- (8.096584372949817,3.7850789370132016) -- (8.171110490266546,3.872897986025821) -- (8.217182993283568,3.9226477593890223) -- (8.269292416839885,3.975689164955388) -- (8.327512444145642,4.0315546461610925) -- (8.391876341633823,4.089767693680587) -- (8.46237646173816,4.1498462248925945) -- (8.538964030171984,4.211305995746733) -- (8.621549219353739,4.27366402117849) -- (8.710001507627362,4.336441980164538) -- (8.804150321930612,4.399169581621942) -- (8.903785959585482,4.461387867631596) -- (8.955585394102203,4.492166485828076) -- (9.00866078293624,4.522652430908643) -- (9.062975483489133,4.5527929356720565) -- (9.118490678655359,4.582536524044534) -- (9.175165436176101,4.611833100806813) -- (9.232956771689922,4.640634038806049) -- (9.291819715371368,4.66889226350014) -- (9.351707382042658,4.696562334687138) -- (9.412571044636916,4.723600525276732) -- (9.474360210886253,4.749964896966658) -- (9.537022703101847,4.775615372691765) -- (9.600504740907905,4.800513805719338) -- (9.664751026786517,4.82462404527028) -- (9.729704834284924,4.847911998551485) -- (9.795308098732743,4.870345689091446) -- (9.861501510311705,4.891895311277224) -- (9.928224609316857,4.912533280997862) -- (9.995415883443773,4.932234282305862) -- (10.063012866933487,4.9509753100156075) -- (10.130952241403142,4.968735708164525) -- (10.1991699381871,4.985497204269971) -- (10.267601242011773,5.001243939322475) -- (10.336180895823578,5.015962493463078) -- (10.404843206588707,5.029641907300327) -- (10.473522151880879,5.042273698829964) -- (10.542151487072264,5.053851875928132) -- (10.610664852941689,5.064372944396572) -- (10.678995883512957,5.07383591154621) -- (10.747078313936868,5.082242285313148) -- (10.814846088229142,5.089596068909095) -- (10.882233466677432,5.095903751015947) -- (10.949175132731316,5.101174291542034) -- (11.015606299189404,5.1054191029654135) -- (11.081462813499407,5.108652027297229) -- (11.146681261988814,5.110889308705717) -- (11.211199072844446,5.112149561849462) -- (11.274954617663363,5.112453735975242) -- (11.33788731139779,5.1118250748441945) -- (11.399937710521575,5.110289072556348) -- (11.461047609247313,5.107873425351542) -- (11.521160133627486,5.104607979471291) -- (11.580219833376313,5.100524675173331) -- (11.638172771253275,5.095657486997011) -- (11.694966609852866,5.090042360384825) -- (11.750550695650697,5.083717144770967) -- (11.804876140160033,5.076721523254451) -- (11.857895898057695,5.069096938980814) -- (11.909564842144807,5.060886518360729) -- (12.00867979728447,5.042888608307571) -- (12.101900987321622,5.023103365909916) -- (12.188943304786273,5.001927901077863) -- (12.269558323501087,4.979777931854013) -- (12.343535816839006,4.9570847201508235) -- (12.410705013261905,4.934291893682732) -- (12.470935580500994,4.9118521763852945) -- (12.570265648237182,4.869868372033166) -- (12.641311884633328,4.834809233942952) -- (12.684522780750644,4.810279982757955) -- (12.700991096128492,4.799708278797393) -- (12.692417854408058,4.80624203926734) -- (12.661059709816808,4.832654791873839) -- (12.609661164622475,4.8812612204065395) -- (12.541373539237734,4.953845274278076) -- (12.459662965173013,5.051602864237218) -- (12.414918737242493,5.1101241980888545) -- (12.368209976228304,5.175100760282182) -- (12.320012849463032,5.246501198323274) -- (12.270803506673493,5.324252856503989) -- (12.221054670987197,5.408242026687467) -- (12.171232260188432,5.498314483772389) -- (12.121792061863461,5.594276302252823) -- (12.097354154833376,5.644394527941492) -- (12.073176485696045,5.695894948328201) -- (12.049311748858877,5.748742687327784) -- (12.025811415629677,5.802900640095599) -- (12.002725642847805,5.858329529518985) -- (11.980103183913624,5.91498796648471) -- (11.95799130237088,5.972832513819151) -- (11.93643568819388,6.031817753790238) -- (11.915480376925498,6.09189635905474) -- (11.895167671807428,6.153019166928295) -- (11.875538069038866,6.2151352568497185) -- (11.856630186294582,6.278192030905056) -- (11.838480694626416,6.342135297273372) -- (11.829701256899856,6.374422015460825) -- (11.827537869059505,6.382525321477765) -- (11.826460939777109,6.38658164183191) -- (11.825923669746974,6.388610962723897) -- (11.825856567062562,6.3888646821156465) -- (11.825823020395077,6.388991546330813) -- (11.825806248230139,6.389054979567998) -- (11.825802055310639,6.3890708379949865) -- (11.825799958869178,6.389078767226088) -- (11.825799434760693,6.38908074953574) -- (11.825799402003943,6.38908087343011)(13.37999999985866,6.799999996264545) -- (13.38000000125866,6.800000033264546) -- (13.38,6.8) -- (13.38,6.8) -- (13.37999999999986,6.8) -- (13.379999999999399,6.8) -- (13.379999999997509,6.8) -- (13.379999999989856,6.8000000000003835) -- (13.379999999959065,6.800000000001549) -- (13.379999999835547,6.800000000006223) -- (13.379999999340763,6.80000000002495) -- (13.379999997360214,6.800000000099927) -- (13.379999989435179,6.8000000004000976) -- (13.37999995772941,6.800000001602208) -- (13.37999983089533,6.800000006420803) -- (13.379999323539261,6.800000025773812) -- (13.379997294093695,6.800000103809719) -- (13.37998917641886,6.800000420933789) -- (13.379956707204608,6.80000172924967) -- (13.3798268446011,6.8000072809722205) -- (13.379307530899695,6.800032034419421) -- (13.377231701950876,6.800151399878186) -- (13.368945035429949,6.800791213872113) -- (13.336014777730815,6.804637745101998) -- (13.27053815638065,6.815906191969983) -- (13.177528852279126,6.837562067249879) -- (13.121354690266097,6.853221448452537) -- (13.059152893587227,6.872554617040203) -- (12.991277101070596,6.895867463402641) -- (12.918109776782082,6.9234359747089425) -- (12.840059594435742,6.9555040434121995) -- (12.757558645254536,6.992281501300751) -- (12.671059488741486,7.033942392982405) -- (12.581032066914496,7.080623501015099) -- (12.487960503507585,7.132423133139209) -- (12.392339810440799,7.189400180235326) -- (12.294672524504708,7.2515734517396355) -- (12.195465297689006,7.318921293310251) -- (12.145442815864278,7.354517816773157) -- (12.095225464904733,7.391381489565456) -- (12.04487620541553,7.429498228056174) -- (11.994457613003657,7.468851452722758) -- (11.944031772809286,7.5094221095785825) -- (11.893660174985866,7.551188695969474) -- (11.843403611310821,7.594127290695415) -- (11.793322073106888,7.638211588405824) -- (11.693919434250796,7.729700386433321) -- (11.595912400427522,7.825398534529168) -- (11.499742056584912,7.925014226464182) -- (11.452477546443838,7.976190763786563) -- (11.405827497103383,8.028222216201147) -- (11.35984040410036,8.081063039038074) -- (11.314563051369323,8.134665865964887) -- (11.27004043458478,8.188981584946564) -- (11.226315687709192,8.24395941731687) -- (11.183430012874302,8.299546999823963) -- (11.141422613717472,8.355690469507909) -- (11.100330632288609,8.412334551262926) -- (11.060189089637138,8.469422647932385) -- (11.021030830181894,8.52689693278018) -- (10.9828864699603,8.58469844417791) -- (10.94578434884635,8.642767182343503) -- (10.909750486820014,8.701042207963246) -- (10.874808544363514,8.75946174252585) -- (10.840979787052845,8.817963270194346) -- (10.808283054405411,8.8764836410387) -- (10.776734733037337,8.93495917544985) -- (10.746348734176383,8.993325769553538) -- (10.717136475568855,9.051519001440886) -- (10.689106867811129,9.109474238030936) -- (10.662266305128746,9.167126742379379) -- (10.63661866061824,9.224411781246936) -- (10.612165285958996,9.281264732740357) -- (10.588905015594705,9.337621193838759) -- (10.566834175376135,9.393417087618497) -- (10.545946595649083,9.448588769989884) -- (10.526233628763693,9.503073135760317) -- (10.507684170973521,9.556807723839155) -- (10.490284688685058,9.609730821401511) -- (10.474019249010933,9.661781566829745) -- (10.458869554572212,9.712900051253524) -- (10.431832627484146,9.812105963363454) -- (10.40898182250572,9.906892095113518) -- (10.390090139486498,9.99682342832729) -- (10.37489748101311,10.081489290304873) -- (10.36311270478009,10.16050609061157) -- (10.354415895885491,10.233519874588234) -- (10.3484608430399,10.300208675827722) -- (10.344877701251868,10.360284651262083) -- (10.343275822248632,10.413495984019523) -- (10.344367239427328,10.498507272454674) -- (10.348310012413405,10.554005011367028) -- (10.35154805892356,10.579406660733575) -- (10.350501752332258,10.574798588633813) -- (10.341676137992048,10.540931748707843) -- (10.321766872465066,10.479199241849903) -- (10.287760911325169,10.391596391393815) -- (10.264632744392477,10.338868253781417) -- (10.237029117461953,10.280664447612738) -- (10.204694821284493,10.217376259999853) -- (10.167408186994834,10.14941949280259) -- (10.124982903534514,10.077231574866365) -- (10.07726958770722,10.001268529850684) -- (10.02415709574086,9.922001820808429) -- (9.965573567036223,9.839915093532154) -- (9.901487192652557,9.755500841387207) -- (9.831906703003174,9.669257014895285) -- (9.756881571195594,9.581683599714108) -- (9.676501930437102,9.493279186875721) -- (9.590898205924123,9.404537559195978) -- (9.500240463628513,9.31594431765092) -- (9.404737480371809,9.227973571233102) -- (9.355244560717972,9.18436585497156) -- (9.304635541525787,9.141084713353395) -- (9.252946828559708,9.098184760194425) -- (9.200216974581004,9.055719307246111) -- (9.146486619804733,9.013740274341913) -- (9.091798428670565,8.972298102070026) -- (9.036197023036673,8.931441667124474) -- (8.979728911911945,8.891218200481688) -- (8.922442417848627,8.85167320854531) -- (8.864387600122122,8.812850397395948) -- (8.805616174831664,8.77479160027805) -- (8.746181432059657,8.737536708449616) -- (8.686138150233658,8.70112360551514) -- (8.62554250783933,8.665588105355823) -- (8.564451992637311,8.630963893764756) -- (8.502925308541636,8.597282473888498) -- (8.441022280321167,8.564573115569623) -- (8.378803756289333,8.53286280867805) -- (8.316331509151171,8.502176220512027) -- (8.253668135179645,8.472535657342457) -- (8.19087695189638,8.44396103016704) -- (8.128021894434326,8.416469824733328) -- (8.065167410762394,8.390077075882331) -- (8.002378355954058,8.364795346256738) -- (7.939719885683381,8.34063470941022) -- (7.877257349133712,8.317602737346492) -- (7.81505618150485,8.295704492509191) -- (7.753181796305434,8.274942524235751) -- (7.69169947761784,8.255316869680644) -- (7.630674272522687,8.236825059205595) -- (7.570170883869674,8.219462126226544) -- (7.510253563581428,8.20322062149927) -- (7.450986006675363,8.188090631817984) -- (7.3924312461880595,8.17405980309323) -- (7.334651549184478,8.161113367768063) -- (7.277708314032785,8.14923417652369) -- (7.221661969123444,8.138402734218243) -- (7.166571873208305,8.12859723999514) -- (7.112496217533041,8.119793631489987) -- (7.059491929932986,8.111965633057876) -- (7.007614581058869,8.105084807936013) -- (6.956918292895492,8.099120614249541) -- (6.859277611740362,8.089809790263699) -- (6.766970572266474,8.083749084449337) -- (6.680369267819115,8.080624930771085) -- (6.599814765128625,8.080096677783715) -- (6.525614881494592,8.081799185704662) -- (6.458042188741521,8.08534559831807) -- (6.397332257791619,8.09033027022082) -- (6.343682156023806,8.096331828833412) -- (6.297249207825851,8.102916349656134) -- (6.226459827086654,8.116055485445616) -- (6.185398075966052,8.126150870192943) -- (6.173829389337995,8.129618899351366) -- (6.190860540199884,8.122997047420226) -- (6.234967474343825,8.103048521596826) -- (6.3040403328915975,8.066860755307275) -- (6.395444401247412,8.01193497547883) -- (6.448572015736051,7.9767971074831) -- (6.5060952799308485,7.936264326031401) -- (6.567572709569544,7.890173219954471) -- (6.632546177199034,7.838398348924068) -- (6.700544128605016,7.780853388821411) -- (6.70488347073723,7.777063969934801) -- (6.7070568303560165,7.7751607334764365) -- (6.708144425810191,7.774206983180291) -- (6.708688451636478,7.773729574978314) -- (6.708960521472894,7.773490737609328) -- (6.709096570609234,7.773371285607279) -- (6.7090986964521715,7.773369418993505) -- (6.709099227913271,7.773368952339216) -- (6.709099493643863,7.773368719011943) -- (6.709099626509184,7.773368602348275) -- (6.709099654891097,7.773368577427262)(11.121535660791613,8.38288446272868);
\draw (11.841901199278373,8.815767217020941)-- (11.121535660791613,8.38288446272868);
\begin{scriptsize}
\draw [fill=uuuuuu] (9.68,6.94) circle (1.5pt);
\draw [fill=qqqqff] (13.38,6.8) circle (2.5pt);
\draw[color=qqqqff] (13.64,7.17) node {$A$};
\draw [fill=uuuuuu] (12.476692997302312,9.366542453541712) circle (1.5pt);
\draw [fill=uuuuuu] (11.841901199278373,8.815767217020941) circle (1.5pt);
\draw[color=black] (10.84,6.45) node {$R$};
\draw [fill=uuuuuu] (11.121535660791613,8.38288446272868) circle (1.5pt);
\draw [fill=qqqqff] (10.350501752332258,10.574798588633813) circle (2.5pt);
\draw[color=qqqqff] (10.62,10.95) node {$B$};
\end{scriptsize}
\end{tikzpicture}
\\
Central Brachistochrone.
\end{center}
\end{example}
\begin{example} \label{ExCC}(Central Catenary)
We can similarly ``fix'' the classical problem of Catenary -- i.e. what is the shape of a freely hanging chain of given length fixed on both ends under influence of gravity. Since the chain is motionless, the quantity to minimize is just the potential energy $U$.

For a classical Catenary the approximation used is, again, $U=-gy$. If we, instead, use $U=-\frac{M}{r}$, the problem of \textit{Central Catenary} can be rephrased now in this way:

\textit{Which shape assumes a chain of given length that is attached to two geostationary satellites?}

Introducing the Lagrange multiplier $\lambda$ to accommodate the fixed length of the rope condition, we end up (using suitable units) with the following functional:
$$
\mathcal{L}:=\inte{s_0}{s_1}\zav{\frac{M}{r}+\lambda }\dd s.
$$ 
Thus by Proposition \ref{P2} we obtain extremal curves given by
$$
\frac{L}{p}=\frac{M}{r}+\lambda.
$$
This result can be, of course, interpreted via Theorem \ref{T2} as a solution of
$$
{\bf \ddot x}=-\frac{M \abs{\bf \dot x}}{r^3}{\bf x},
$$
but this time we can do better than this.

Squaring both sides we get
$$
\frac{L^2}{p^2}=\frac{M^2}{r^2}+\frac{2\lambda M}{r}+\lambda^2.
$$
Without the $\frac{M^2}{r^2}$ this would represent a hyperbola (since the energy $\lambda^2$ is positive).

But we can get rid of this term using harmonics $H_k$:
$$
\frac{L^2k^2}{p^2}=\frac{M^2+L^2(k^2-1)}{r^2}+\frac{2\lambda M}{r}+\lambda^2,
$$
Thus for $L^2\geq M^2$ the central Catenary is $\sqrt{1-\frac{M^2}{L^2}}$-harmonic of an hyperbola!

Exactly the same result can be obtained noticing that central Catenary is Dual to a curve with linear pedal equation:
$$
\frac{L}{p}=\frac{M}{r}+\lambda \qquad \stackrel{D}{\longrightarrow} \qquad Lr=Mp+\lambda,
$$
which we already know is a special case of Dual Harmonics of a circle that does not contain the origin. But Dual of Dual Harmonics are jut Harmonics of a Dual: $D H^\star_k=D D H_k D= H_k D$. Thus we end up with Harmonics of Dual of an outside circle, i.e. Harmonics of a hyperbola.
\end{example}
\begin{example}(Dark Catenary)
We can, of course, consider various other formulas for the potential energy $U$ and compute the corresponding Catenary. Interesting choice is
$$
U=r^2,
$$
which would generate a variational problem that can be phrased as follows:

\textit{Tunnel of what shape we must dig through Earth so that a rope of given length can be hanged freely inside it without touching the walls?} 

The functional to minimize now takes form
$$
\mathcal{L}:=\inte{s_0}{s_1}\zav{r^2+\lambda }\dd s,
$$ 
with solution
$$
\frac{L}{p}=r^2+\lambda.
$$
Again, we can in no time produce corresponding dynamical system:
$$
{\bf \ddot x}=2\abs{\bf \dot x}{\bf x}.
$$

If we try to compute the Dual curve we obtain
$$
\frac{L}{p}=r^2+\lambda, \qquad \stackrel{D}{\longrightarrow} \qquad \frac{1}{Lp^2}=r-\frac{\lambda}{L},
$$
which can reinterpreted as a solution to
$$
{\bf \ddot x}=\frac12 \frac{{\bf x}}{r},
$$
i.e. evolution of a particle that experiences constant outward repulsing force -- a.k.a. dark energy. This might be seen as a special case of the so-called ``Dark Kepler problem'' introduced in \cite{Blaschke6}, where in addition to dark energy, influences of dark matter and regular matter are assumed.

Thus Dual of Dark Catenary (or Dark co-Catenary) gives us a particular solution to the Dark Kepler problem! 
\end{example}
\newpage
\section{Dipole drive}
Robert Zubrin came with the idea of ``Dipole drive'' \cite{Zubrin} -- i.e. a spacecraft that harness the solar wind for propulsion. Similarly to a Solar sail, Dipole drive does not carry any fuel but unlike a solar sail, it needs its own source of energy -- which would be a small on board nuclear generator. 

Dipole drive architecture calls for two large planar parallel meshes, separated by certain distance -- both of them are charged. One with positive charge and the other with negative charge of equal intensity. Together they generate an electric field $E$ between them which we will presume is uniform -- an approximation valid only for infinitely large screens. Similarly, we will assume that the net outside electric field is zero. 

Charged particles present in the solar wind would therefore enter the area between the screens unobstructed (we assume that the gaps are so large that any collisions might be ignored) and will be accelerated by the present electric field. The protons will be pushed towards the negatively charged screen and electrons the other way around. The energy change for both particles is the same but since a proton is 1836 times more massive it will gain about 42 times the momentum of an electron. Thus a net thrust is generated. The resulting force on a spacecraft is in the direction of the electric filed $E$, which is perpendicular to the screens.

The following picture illustrates the situation. 

\begin{center}
\includegraphics[scale=1]{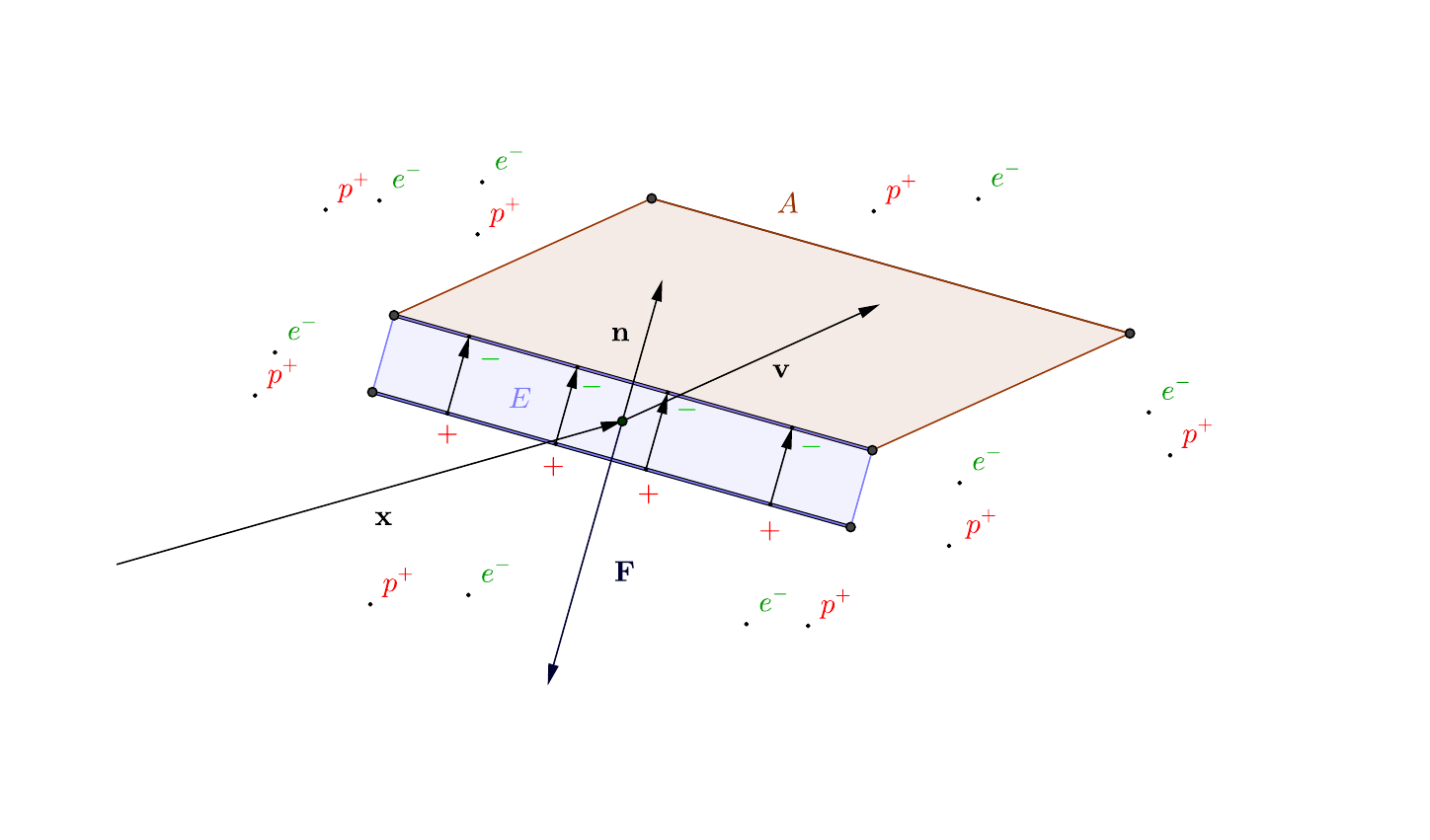}
\end{center}

In order to derive equation of motion of the Dipole drive we will shift to a moving frame of reference of the solar wind near the spacecraft -- i.e. in this frame protons and electrons are (locally) motionless, while the spacecraft is moving with the velocity ${\bf v}$. The amount of thrust $\abs{\bf F}$ is proportional to the density of solar wind $\rho(r)$ which we will assume depends only on the distance from the Sun $r$ and on the area $A$ which spacecraft ``scoops'' as it goes. This area, in turn, is proportional to the dot product of velocity vector ${\bf v}$ and the screen normal vector ${\bf n}$. Hence we get:
$$
\abs{\bf F}=\alpha \rho(r) {\bf v}\cdot {\bf n}.
$$
(We are assuming that both vectors ${\bf n}, {\bf v}$ lies in the same plane as ${\bf x}$ does. This will ensure that the resulting orbit is a planar curve and thus suitable for analysis in pedal coordinates.)

The equation of motion is thus
\begin{equation}\label{DDE1}
{\bf \ddot x}=-\frac{M}{r^3}{\bf x}+\alpha \rho(r) ({\bf v}\cdot {\bf n})\ {\bf n}.
\end{equation}

The velocity vector ${\bf v}$ comprises of two components -- the spacecraft's own speed ${\bf \dot x}$ (relative to the Sun) and the velocity of solar wind particles themselves -- which in our case are moving radially away form the Sun with speed $v(r)$ also depending only on the distance $r$.
$$
{\bf v}:={\bf \dot x}+v(r)\frac{{\bf x}}{r}.
$$
Substituting in (\ref{DDE1}) we get
\begin{equation}\label{DDE2}
{\bf \ddot x}=-\frac{M}{r^3}{\bf x}+\alpha \rho(r) \zav{{\bf \dot x}\cdot {\bf n}+\frac{v(r)}{r}{\bf x}\cdot {\bf n}}\ {\bf n}.
\end{equation}

The standard approximation to the values of $v(r),\rho(r)$ are
$$
v(r)\propto 1,\qquad \rho(r) \propto \frac{1}{r^2},
$$
which is valid only a certain distance away from Sun. We will only need a much more conservative assumption:
$$
v(r)\rho(r)\propto  \frac{1}{r^2},
$$
which has been observed from numerous measurements (see \cite{solarwind}).

\subsection{Orbits of Dipole drive pointing in direction perpendicular to motion}
Substituting 
$$
{\bf n}:=\frac{{\bf \dot x}^\perp}{\abs{\bf \dot x}},
$$
in (\ref{DDE2}) we get
\begin{equation}\label{DDE3}
{\bf \ddot x}=-\frac{M}{r^3}{\bf x}+\alpha \rho(r)v(r)\frac{{\bf x}\cdot {\bf \dot x}^\perp}{r\abs{\bf \dot x}^2}{\bf \dot x}^\perp=
-\frac{M}{r^3}{\bf x}+\alpha\frac{{\bf x}\cdot {\bf \dot x}^\perp}{r^3\abs{\bf \dot x}^2}{\bf \dot x}^\perp,
\end{equation}
which can be written using definition of $p$ as follows:
$$
{\bf \ddot x}=-\frac{M}{r^3}{\bf x}-\frac{\alpha p}{r^3} \frac{\dot {\bf x}^\perp}{\abs{\dot{\bf x}}}.
$$
Since the thrust $\abs{\bf F}$ of our Dipole drive in this mode falls as a square of the distance (exactly the same as the force of gravity) we will write the constant of proportion $\alpha$ as follows:
$$
\alpha=\sigma M,
$$  
which signifies how the force compares with gravity. If $\abs{\sigma}<1$ the thrust is weaker. If $\abs{\sigma}=1$ it is exactly equal the force of gravity and so on.

Finally, we obtain:
\begin{equation}\label{DDE4}
{\bf \ddot x}=-\frac{M}{r^3}{\bf x}-\frac{\sigma M p}{r^3} \frac{\dot {\bf x}^\perp}{\abs{\dot{\bf x}}}.
\end{equation}

This is exactly in the usable form for Theorem \ref{T2} and thus using (\ref{T2F3}) we obtain: 
$$
\frac{-2\sigma M \int\frac{p^2}{r^2}{\rm d}r}{p^2}=\frac{2M}{r}+c.
$$
Differentiating out the integral we obtain a separable differential equation with the solution:
$$
\frac{L^2}{p^2}=\zav{\frac{2M}{r}+c}^{1-\sigma}.
$$
 
Particular solutions include:
\begin{enumerate}
\item For $\sigma=0$ (i.e. screens are without charge and no force is generated) we obtain the usual Kepler problem whose solutions are focal conic sections. 
\item For $\sigma=1$ we obtain the equation of a line $p=L$. In fact, the trajectory is a line segment which the Dipole drive traverse in periodical motion from one end to the other. But there is no way to distinguish between the two in pedal coordinates.
\item For $\sigma=-1$ we obtain 
$$
\frac{L^2}{p^2}=\zav{\frac{2M}{r}+c}^2.
$$
which is precisely the central Catenary (Example \ref{ExCC})!
\item For $\sigma=2$ (i.e. the thrust generated is exactly twice as big as the gravity force) we obtain
$$
\frac{L^2}{p^2}=\zav{\frac{2M}{r}+c}^{-2}, \qquad \Rightarrow\qquad \frac{L}{p}=\frac{1}{\sqrt{\frac{2M}{p}+c}},
$$
which is \textit{exactly} central Brachistochrone (Example \ref{CBex})!
\item For $\sigma=\frac12$ we obtain 
$$
\frac{L^2}{p^2}=\sqrt{\frac{2M}{r}+c}, \qquad \Rightarrow\qquad
\frac{L^4}{p^4}=\frac{2M}{r}+c,
$$
which can be reinterpreted back using Theorem \ref{T2} as a solution to the equation
$$
{\bf \ddot x}=-\frac{M}{r^3} \frac{{\bf x}}{\abs{\bf \dot x}^2},
$$
i.e. to the force problem when the force varies not only inversely as a square of the distance but also inversely as a square of the speed.

\end{enumerate}

Of course the last case can be generalized to all the values $\sigma$ but $\sigma\not =1$:
$$
\frac{L^2}{p^2}=\zav{\frac{2M}{r}+c}^{1-\sigma} \qquad \Rightarrow\qquad \zav{\frac{L}{p}}^{\frac{2}{1-\sigma}}=\frac{2M}{r}+c,
$$
which s equivalent (by Theorem \ref{T2}) to solutions of
$$
{\bf \ddot x}=-\frac{M}{r^3}\frac{{\bf x}}{\abs{\bf \dot x}^{\frac{2\sigma}{1-\sigma}}}.
$$ 
\subsection{Complementary system}
Interestingly, if we replace in (\ref{DDE3}) the vector ${\bf \dot x}^\perp$ by ${\bf \dot x}$ we obtain a dynamical system of the form:
$$
\ddot {\bf x}=-\frac{M}{r^3}{\bf x}+\frac{\sigma M p_c}{r^3} \frac{{\bf \dot x}}{\abs{\bf \dot x}}.
$$
This has nothing to do with trajectories of a Dipole drive. But using  $\dot {\bf x}=\frac{\abs{\dot {\bf x}}}{p_c} {\bf x}+\frac{p}{p_c}{\dot {\bf x}}^\perp$
we obtain
$$
\ddot {\bf x}=-\frac{M(1-\sigma)}{r^3}{\bf x}+\frac{\sigma M p}{r^3} \frac{\dot {\bf x}^\perp}{\abs{\dot{\bf x}}},
$$
which is exactly the same equation! Hence using Theorem \ref{T2} we get nearly the same curves:
$$
\frac{L^2}{p^2}=\zav{\frac{2M(1-\sigma)}{r}+c}^{\frac{1}{1-\sigma}}.
$$

\subsection{Connections} The example of Dipole drive shows how many connections between various dynamical systems can be deduced from Theorem \ref{T2}. Taking, for instance, $\sigma=-1$ in (\ref{DDE4}), we can establish that the following list of dynamical systems have solutions with same orbits (as far as pedal coordinates are concerned, that is):
\begin{align*}
{\bf \ddot x}&=-\frac{M }{r^3}{\bf x}+\frac{Mp}{r^3}\frac{{\bf \dot x}^\perp}{\abs{\bf \dot x}},\\
{\bf \ddot x}&=-\frac{2M }{r^3}{\bf x}+\frac{Mp_c}{2r^3}\frac{{\bf \dot x}}{\abs{\bf \dot x}},\\
{\bf \ddot x}&=-\frac{M }{r^3}\abs{\bf \dot x} {\bf x}.
\end{align*}
Additionally, these curves also solves the central Catenary problem!
\newpage
\section{Solar Sail orbits}
Solar sail, as a concept, was first mentioned in a note send by Johannes Kepler to Galileo in 1610. Kepler has noticed that the tail of a comet always points directly away from the Sun as if blown away by some ``heavenly breeze''. He then suggests to harness this breeze as a means of propulsion. Later it was  recognized that the source of this force is the pressure of the light itself. 

For our purposes, a solar sail will be a perfectly planar object (a sail) that is perpendicular to the plane of motion - in other words, the normal vector of the sail is located in the plane of motion. The single most important quantity we can attached to a solar sail is its specific area $A_{sp}:=\frac{A}{m}$ (i.e. reflective area divided by its total mass). 

Since the pressure of the light falls with a square of the distance -- exactly the same as the force of gravity -- there is a unique value of $A_{sp}$ where the force generated by the sail exactly cancels the gravity of Sun. 

Sufficiently large and lightweight solar sails thus posses unique ability among almost all other spacecrafts that they can ``hover'' above the Sun or move in a straight line relative to the Sun -- mode of transportation unheard of in celestial mechanics. 

The quantity $A_{sp}$ will enter our calculation as a ratio $\sigma:=\frac{F_{reflective}}{F_{gravity}}$ of reflective and gravity forces -- which is called ``lightness number''. 

Value $\sigma=1$ thus indicates that a spacecraft in question posses the ``hover'' capability, while $\sigma=0$ corresponds to a spacecraft with no sail. And so on.


\begin{center}
\definecolor{yqqqqq}{rgb}{0.5019607843137255,0.,0.}
\definecolor{qqqqff}{rgb}{0.,0.,1.}
\definecolor{ffdxqq}{rgb}{1.,0.8431372549019608,0.}
\definecolor{cczzqq}{rgb}{0.8,0.6,0.}
\definecolor{ttqqqq}{rgb}{0.2,0.,0.}
\definecolor{uuuuuu}{rgb}{0.26666666666666666,0.26666666666666666,0.26666666666666666}
\begin{tikzpicture}[line cap=round,line join=round,>=triangle 45,x=1.0cm,y=1.0cm,scale=0.7]
\clip(1.82,1.88) rectangle (14.56,9.72);
\draw [line width=2.pt,color=ttqqqq] (8.04,2.66)-- (12.04,7.32);
\draw [line width=0.8pt,dash pattern=on 3pt off 3pt,domain=1.82:14.56] plot(\x,{(-63.4134--4.*\x)/-4.66});
\draw [line width=1.6pt,color=ffdxqq,domain=1.8200000000000032:10.040000000000001] plot(\x,{(-75.45117113343022--6.800743352883156*\x)/-1.4372160061088817});
\draw [->,line width=2.pt,color=yqqqqq] (10.04,4.99) -- (11.908270105738739,3.3863346731856314);
\draw [->,line width=1.6pt,color=ffdxqq] (10.04,4.99) -- (3.1,5.38);
\begin{scriptsize}
\draw [fill=uuuuuu] (8.04,2.66) circle (1.5pt);
\draw [fill=uuuuuu] (12.04,7.32) circle (1.5pt);
\draw [fill=uuuuuu] (10.04,4.99) circle (1.5pt);
\draw[color=uuuuuu] (11.54,5.35) node {Solar sail};
\draw [fill=ffdxqq] (3.1,5.38) circle (2.5pt);
\draw[color=ffdxqq] (3.78,5.89) node {To the sun};
\draw [fill=uuuuuu] (10.04,4.99) circle (1.5pt);
\draw [fill=qqqqff] (8.60278399389112,11.790743352883158) circle (2.5pt);
\draw[color=ffdxqq] (10.3,9.51) node {reflected light};
\draw[color=yqqqqq] (11.2,4.43) node {$F$};
\end{scriptsize}
\end{tikzpicture}
\end{center}

Let us denote ${\bf x}$ the position of the center of mass of a solar sail, ${\bf \dot x}$ its velocity and ${\bf n}$ the normal vector of the sail chosen so that it points \textit{away} from the Sun -- that is the quantity $\cos\eta:=\frac{{\bf x}\cdot {\bf n}}{\abs{{\bf x}}}\geq 0$.

The force generated by the sail:
\begin{align*}
F_{r}&\propto \frac{\zav{\cos \eta}^2}{r^2}, 
\end{align*}
is proportional to $\cos^2\eta$, since the apparent area of the sail as seen from the Sun goes as $\cos\eta$ and also only a portion of the light's momentum is transfered to the craft (which introduces the second $\cos\eta$ factor). For more details see \cite{solarsail}. 

Equations of motion are:
\begin{align}
\label{reflective}\ddot {\bf x}&=-\frac{M}{r^3} {\bf x}+\frac{\sigma M \zav{{\bf x}\cdot {\bf n}}^2}{r^4} {\bf n}. 
\end{align}
The quantity $M$ is the mass of the Sun multiplied by gravity constant $G$.

We will handle 3 different choices for ${\bf n}$:
\begin{align}
\label{prograde}{\bf n}&= \frac{ {\bf \dot x}}{\abs{ {\bf \dot x}}}s.  & s&:=\sgn\zav{\frac{{\bf x}\cdot{\bf \dot x}}{r\abs{\bf \dot x}}} & \text{Prograde direction}\\
\label{normal}{\bf n}&= \frac{ {\bf \dot x}^\perp}{\abs{ {\bf \dot x}}}s, & s&:=\sgn\zav{\frac{{\bf x}\cdot{\bf \dot x}^\perp}{r\abs{\bf \dot x}}}  & \text{Normal direction}\\
\label{constantangle}{\bf n}&=  \frac{\cos\alpha {\bf x}+\sin\alpha {\bf x}^\perp}{\abs{\dot {\bf x}}}. & &  & \text{Constant angle}
\end{align}

The $s$ factor is present  to ensure that the thrust vector of the sail points always away from the Sun (since a solar sail cannot be pushed toward the Sun by the light it emits).

\subsection{Orbits of reflective  solar sail pointing pointing in normal direction}
Substituting (\ref{normal}) into (\ref{reflective}) we get
$$
\ddot {\bf x}=-\frac{M}{r^3}{\bf x}+\frac{\sigma M}{r^2} \frac{\dot {\bf x}^\perp}{\abs{\dot{\bf x}}}\zav{\frac{{\bf x}\cdot \dot {\bf x}^\perp}{r \abs{\dot{\bf x}}}}^2 s^3.
$$
The equation becomes singular for zero speed ${\bf \dot x}=0$. The trajectory must be therefore computed piecewise between such events. We fill focus on one such piece which allows us to drop the $s^3$ factor. Using also the definition of pedal coordinate $p$ we obtain:
$$
\ddot {\bf x}=-\frac{M}{r^3}{\bf x} +\frac{\sigma M p^2}{r^4} \frac{\dot {\bf x}^\perp}{\abs{\dot{\bf x}}}.
$$
Applying Theorem \ref{T2} we get:
$$
\frac{\sigma M \int \frac{p^3}{r^3}{\rm d}r}{p^2}=\frac{M}{r}+c.
$$
Differentiating out the integral we obtain a differential equation with solution:
$$
\frac{L}{p}=\sqrt{c+\frac{M}{r}}+\frac{\sigma L}{r}+\frac{2\sigma c L}{M}.
$$
This can be viewed as a special energy solution to 
$$
{\bf \ddot x}=-\zav{\frac{M}{\sqrt{c+\frac{M}{r}}}+\tilde \sigma}\frac{\abs{\bf \dot x}}{r^3}{\bf x},\qquad \tilde \sigma:=L\sigma,
$$
provided 
$$
E:=\abs{\dot x}-\sqrt{c+\frac{M}{r}}-\frac{\tilde \sigma}{r}=\frac{2\tilde\sigma c}{M}.
$$

The same equation can be also viewed as a solution to a variational problem of finding an extremum of
$$
\inte{s_0}{s_1}\sqrt{c+\frac{M}{r}}\dd s,\qquad \text{under constrain}\qquad
\inte{s_0}{s_1}\zav{\frac{1}{r}+\frac{2c}{M}}\dd s=l.
$$
Notice that we have interpreted the quantity $\sigma L$ as a Lagrange multiplier.
\bigskip

The other choices for ${\bf n}$ leads to equation that cannot be solved easily in pedal coordinates. In the rest of the paper we will therefore present a possible way how to deal with this situation. 
\newpage
\section{Nonlocal variables}
Let us define for every point ${\bf x}$ on a planar curve $\gamma$ the following  quantities:
\begin{align*}
x&:={\bf x}\cdot (1,0),& y&:={\bf x}\cdot (0,1)&\text{cartesian coordinates,}\\
r&:=\abs{\bf x} & \varphi&:=\arctan \frac{y}{x} &\text{polar coordinates,}\\
p&:=\frac{{\bf \dot x}\cdot {\bf x}^\perp}{\abs{\bf \dot x}} & p_c&:=\frac{{\bf \dot x}\cdot {\bf x}}{\abs{\bf \dot x}} & \text{pedal coordinates,}\\
\theta&:=\arctan (y'_x) & \psi&:=\theta-\varphi=\arctan\frac{p}{p_c} & \text{tangential and polar tangential angle,}\\
\kappa&:=\frac{{\bf \ddot x}\cdot {\bf \dot x}^\perp}{\abs{\bf \dot x}^3} & \rho &:=\frac{1}{\kappa} &  \text{curvature and radius of curvature,}\\
A&:=\frac12\int r^2\dd \varphi & s&:=\int \sqrt{r^2+{r'_\varphi}^2}\dd \varphi & \text{area and arc-length,}
\end{align*}
The following picture roughly illustrates the definitions:
\begin{center}
\definecolor{qqwuqq}{rgb}{0.,0.39215686274509803,0.}
\definecolor{ccqqqq}{rgb}{0.8,0.,0.}
\definecolor{uuuuuu}{rgb}{0.26666666666666666,0.26666666666666666,0.26666666666666666}
\definecolor{qqttqq}{rgb}{0.,0.2,0.}
\definecolor{qqccqq}{rgb}{0.,0.8,0.}
\definecolor{qqqqcc}{rgb}{0.,0.,0.8}
\definecolor{qqqqff}{rgb}{0.,0.,1.}
\begin{tikzpicture}[line cap=round,line join=round,>=triangle 45,x=5.0cm,y=5.0cm]
\clip(-1.3608126664560387,-0.4398735533344589) rectangle (0.7217295117430589,0.9541316445386482);
\draw [shift={(0.,0.)},line width=0.8pt,color=qqwuqq,fill=qqwuqq,fill opacity=0.10000000149011612] (0,0) -- (0.:0.08465618610565437) arc (0.:137.7624531741471:0.08465618610565437) -- cycle;
\draw [shift={(-1.1344354204835139,0.)},line width=0.8pt,color=qqwuqq,fill=qqwuqq,fill opacity=0.10000000149011612] (0,0) -- (0.:0.08465618610565437) arc (0.:38.62180036381115:0.08465618610565437) -- cycle;
\draw [shift={(-0.530986145108566,0.48210326687752686)},line width=0.8pt,color=qqwuqq,fill=qqwuqq,fill opacity=0.10000000149011612] (0,0) -- (-42.2375468258529:0.08465618610565437) arc (-42.2375468258529:38.621800363811204:0.08465618610565437) -- cycle;
\draw [line width=0.8pt,color=qqqqcc,domain=-1.3608126664560387:0.7217295117430589] plot(\x,{(--5.954347205620074--5.2487317463009395*\x)/6.569844235773278});
\draw [shift={(0.,0.)},line width=0.8pt,color=qqccqq,fill=qqccqq,fill opacity=0.10000000149011612]  (0,0) --  plot[domain=1.926375257564634:2.9483484469905936,variable=\t]({0.9245444493130677*0.8824301163685412*cos(\t r)+-0.3810742201256813*0.6913162486860159*sin(\t r)},{0.3810742201256813*0.8824301163685412*cos(\t r)+0.9245444493130677*0.6913162486860159*sin(\t r)}) -- cycle ;
\draw [shift={(0.,0.)},line width=2.pt,color=qqttqq]  plot[domain=1.414163194714933:2.9483484469905936,variable=\t]({0.9245444493130677*0.8824301163685412*cos(\t r)+-0.3810742201256813*0.6913162486860159*sin(\t r)},{0.3810742201256813*0.8824301163685412*cos(\t r)+0.9245444493130677*0.6913162486860159*sin(\t r)});
\draw [line width=0.8pt] (-0.530986145108566,0.48210326687752686)-- (0.,0.48210326687752686);
\draw [line width=0.8pt] (0.,0.48210326687752686)-- (0.,0.);
\draw [line width=0.8pt,dash pattern=on 3pt off 3pt,color=ccqqqq] (0.1232875152316123,-0.33685192407079584) circle (5.241091554714864cm);
\draw [->,line width=0.8pt] (-0.530986145108566,0.48210326687752686) -- (-0.1299963255537937,0.8024591109903051);
\draw [->,line width=0.8pt,color=ccqqqq] (-0.530986145108566,0.48210326687752686) -- (0.1232875152316123,-0.33685192407079584);
\draw [line width=0.8pt,color=qqqqcc] (-0.44197242775452794,0.5532174527494554)-- (0.,0.);
\draw [line width=0.8pt,color=qqqqcc] (-0.08865647156446121,-0.07156136035364863)-- (0.,0.);
\draw [line width=0.8pt] (-0.530986145108566,0.48210326687752686)-- (0.,0.);
\draw [line width=0.8pt,domain=-1.3608126664560387:0.7217295117430589] plot(\x,{(-0.-0.*\x)/2.1344354204835136});
\begin{scriptsize}
\draw [fill=qqqqff] (0.,0.) circle (2.5pt);
\draw [fill=black] (-0.8512525492985398,-0.20726689484536898) circle (0.5pt);
\draw [fill=black] (-0.530986145108566,0.48210326687752686) circle (2.5pt);
\draw[color=black] (-0.5824695341013797,0.5127554302628324) node {${\bf x}$};
\draw[color=qqccqq] (-0.5509351527119452,0.09487135556625731) node {$A$};
\draw [fill=black] (-0.13295114927282456,0.6837842815309425) circle (0.5pt);
\draw[color=qqttqq] (-0.7936162195481544,0.20774627037379634) node {$s$};
\draw [fill=uuuuuu] (0.,0.48210326687752686) circle (1.5pt);
\draw[color=black] (-0.2207760268998931,0.5127554302628324) node {$x$};
\draw[color=black] (0.052945641508389374,0.2133900161141733) node {$y$};
\draw [fill=black] (0.1232875152316123,-0.33685192407079584) circle (0.5pt);
\draw[color=black] (0.25894236103214835,-0.3368751935725795) node {${\bf x''}$};
\draw [fill=black] (-0.1299963255537937,0.8024591109903051) circle (0.5pt);
\draw[color=black] (-0.1389417136644272,0.8370239204258264) node {${\bf x'}$};
\draw[color=ccqqqq] (-0.19790111209235392,-0.10548161821712448) node {$\rho=\frac{1}{\kappa}$};
\draw [fill=black] (-0.44197242775452794,0.5532174527494554) circle (0.5pt);
\draw[color=qqqqcc] (-0.1897354253278198,0.2952243293496391) node {$p$};
\draw [fill=black] (-0.08865647156446121,-0.07156136035364863) circle (0.5pt);
\draw[color=qqqqcc] (0.054356577943483614,-0.12382379187334959) node {$p_c$};
\draw[color=black] (-0.21513228115951613,0.1315557028787075) node {$r$};
\draw[color=qqwuqq] (0.1291362090034783,0.0722963726047495) node {$\varphi$};
\draw [fill=uuuuuu] (-1.1344354204835139,0.) circle (1.5pt);
\draw[color=qqwuqq] (-0.99631754816681437,0.027146406681733876) node {$\theta$};
\draw[color=qqwuqq] (-0.3900883991112018,0.425030481378309) node {$\psi$};
\end{scriptsize}
\end{tikzpicture}
\end{center}
Notice that quantities $s,A$ are written as indefinite integrals which corresponds to the freedom of choice of  the starting point from which the arc-length resp. area is measured. These quantities represents what we will call ``nonlocal variables'' since their value depends not only on the immediate vicinity of the point ${\bf x}$ but also on the curve's ``history''.

Number of these variables can be written in pedal coordinates as nonlocal coordinates (i.e. indefinite integrals):
\begin{lemma}
With the notation as above we have
\begin{align*}
\varphi&=\int\frac{p}{rp_c}\dd r, & \theta&=\int\frac{1}{p_c}\dd p, & s&=\int\frac{r}{p_c}\dd r,\\
A&=\frac12\int \frac{rp}{p_c}\dd r, & \int r^\alpha \dd \varphi &= \int \frac{r^{\alpha-1}p}{p_c}\dd r.
\end{align*}
\end{lemma} 
\begin{proof}
This stems from the fact
$$
r'\varphi=\frac{rp_c}{p},
$$
see Proposition \ref{P2}.
\end{proof}
Also, using the Dual transform 
$$
D:\qquad {\bf x}\to \tilde {\bf x}:= \frac{{\bf x}'}{{\bf x}'\cdot {\bf x}^\perp},
$$ 
where the differentiation is with respect to the arc-length $s$, we can define for every variable $q$ its dual counterpart $q^\star$:
\begin{lemma}
\begin{align*}
x^\star&:=\frac{x'}{p}, & y^\star&:=\frac{y'}{p}, & r^\star&:=\frac{1}{p},\\
p^\star&:=\frac{1}{r}, & p_c^\star&:=-\frac{p_c}{rp}, & \varphi^\star&:=\theta, \\
\theta^\star&:=\varphi, & \psi^\star&:=\psi, & \kappa^\star&:=\zav{\frac{p}{r}}^3\rho, \\
\rho^\star &:=\zav{\frac{r}{p}}^3\kappa,  & A^\star&:=\frac12\int \frac{1}{p_c p^2}\dd p,  & s^\star&:=\int \frac{r}{p^2 p_c}\dd p.  
\end{align*}
\end{lemma} 
\begin{proof} The proof is not difficult and it is left to the reader as an exercise.
\end{proof}

With these quantities in hand there is a hope that for some dynamical systems, the definite integrals in Theorem \ref{T2} can be \textit{interpreted} rather than computed -- similarly as we have interpreted the Euler-Lagrange equation for the elastic curve in the very beginning.
\begin{example} (Aerobraking) Consider a dynamical system:
\begin{equation}\label{aerobraking}
\ddot {\bf x}=-\frac{M}{r^3}{\bf x}-\alpha \abs{\dot{\bf x}}^2 \frac{\dot{\bf x}}{\abs{\dot{\bf x}}},
\end{equation}
i.e. Kepler problem with air drag present. This would correspond to the movement of a spacecraft, that is ``aerobraking'' that is using atmosphere of a planet to slow itself down (assuming that the atmosphere throughout the maneuver is of homogeneous density).

Using 
$$
\dot {\bf x}=\frac{\abs{\dot {\bf x}}}{p_c} {\bf x}+\frac{p}{p_c}{\dot {\bf x}}^\perp,
$$
we would find that this problem is actually \textit{not} suitable for Theorem \ref{T2}.

But! We can proceed differently: multiplying (\ref{aerobraking}) by ${\bf \dot x}^\perp$ we obtain
\begin{align*}
{\bf \ddot x}\cdot {\bf \dot x}^\perp&=-\frac{M}{r^3}{\bf x}\cdot {\bf \dot x}^\perp\\
\kappa\abs{\bf \dot x}^3&=\frac{M}{r^3}p \abs{\bf \dot x}\\
\kappa&=\frac{M}{r^3}\frac{p}{\abs{\bf \dot x}^2}=M\frac{ p^3}{r^3}\frac{1}{\zav{{\bf x}^\perp\cdot {\bf \dot x}}^2}\\
\zav{{\bf x}^\perp\cdot {\bf \dot x}}^2&=M\zav{\frac{p}{r}}^3\rho=M\kappa^\star.
\end{align*}
Similarly, multiplying (\ref{aerobraking}) by ${\bf x}^\perp$ we have
\begin{align*}
{\bf \ddot x}\cdot {\bf x}^\perp&=-\alpha \abs{\bf \dot x}{\bf \dot x}\cdot {\bf x}^\perp\\
\frac{\partial_t\zav{{\bf \dot x}\cdot {\bf x}^\perp}}{{\bf \dot x}\cdot {\bf x}^\perp}&=-\alpha \abs{\bf \dot x}\\
\intertext{using the obvious $\dot s=\abs{\bf \dot x}$ we can integrate this to obtain}
{\bf \dot x}\cdot {\bf x}^\perp&=Le^{-\alpha s},
\end{align*}
where $L$ is the integration constant. Substituting into the previous result we finally get
\begin{equation}\label{aerosol}
M\kappa^\star=Le^{-2\alpha s}.
\end{equation}

This is thus the solution written in coordinates $(\kappa^\star,s)$. It roughly informs us that the curvature of the dual curve (which is, sadly, \textit{not} the curve of velocities in this case) decays exponentially with the arc-length.

Notice that these coordinate looks almost exactly like  ``Cesaro coordinates'' $(\kappa,s)$ which uses only intrinsic parameters of the curve independent on the outside frame of reference.

But intrinsic coordinates are useless in our context. We cannot hope to solve any dynamical system using only intrinsic quantities since the position of the origin (not an intrinsic parameter) is of chief importance. For example, intrinsic coordinates cannot differentiate between Kepler problem and Hook's law (since for both the solutions are conic sections).

But since we have obtained the \textit{dual} curvature $\kappa^\star$, which is not intrinsic and does depend on the position of the origin, there is no problem. Hence, the coordinates $(\kappa^\star, s)$ are closest thing to intrinsic coordinates we can ask for.

Conic sections enjoy particularly compact equation in these coordinates:
$$
\kappa^\star=c.
$$

Furthermore, a different form of the solution (\ref{aerosol}) is
$$
Me^{2\alpha s}=L\rho^\star.
$$
Integrating both sides with respect to $s$ we obtain:
\begin{equation}\label{aerosol2}
\frac{M}{2\alpha}e^{2\alpha s}+c=L\int \frac{r^3}{p^3 p_c}\dd p.
\end{equation}

This stems from the identity
$$
\int \rho^\star \dd s= \int \frac{r^3}{p^3}\kappa \frac{r}{p_c}\dd r=\int \frac{r^3}{p^3 p_c}\dd p,
$$
where we have use famous result
$$
\kappa=\frac{p'_r}{r},\qquad \rho=rr'_p,
$$
which can be readily verified from definitions of $\kappa,\rho, p,r$.

The quantity $\int \frac{r^3}{p^3 p_c}\dd p$ cannot be so easily interpreted as $\kappa^\star$ (it is a running total of dual of radius of curvature) but the equation (\ref{aerosol2}) has two free parameters present -- same number as pedal equations in Theorem \ref{T1} -- so in a sense, it is ''solved more fully''.

\end{example}
\begin{example} (Aerobraking+)
We can modify the previous example and allow the density of the atmosphere to depend on some power of $r$, i.e. we are considering now dynamical system in the form:
\begin{equation}
\ddot {\bf x}=-\frac{M}{r^3}{\bf x}-\alpha r^\beta \abs{\dot{\bf x}} \dot{\bf x}.
\end{equation}
In exactly the same way we will end up with solutions:
$$
\frac{M}{2\alpha}e^{2\alpha \int \frac{r^{\beta+1}}{p_c} \dd r}+c=L\int \frac{r^{\beta+3}}{p^3 p_c}\dd p.
$$
\end{example}

\begin{example}(Solar sail pointing prograde)
When a Solar sail is pointing prograde the equation of motion become
\begin{equation}\label{SSP}
\ddot {\bf x}=-\frac{M}{r^3}{\bf x}+\frac{\sigma M}{r^2} \zav{\frac{{\bf x}\cdot{\dot{\bf x}}}{\abs{\dot{\bf x}}r}}^2 \frac{\dot{\bf x}}{\abs{\dot{\bf x}}}.
\end{equation}
Using 
$$
\dot {\bf x}=\frac{\abs{\dot {\bf x}}}{p_c} {\bf x}+\frac{p}{p_c}{\dot {\bf x}}^\perp,
$$
This can be brought into a form suitable for Theorem \ref{T2} to obtain
\begin{equation}\label{SSPsolt2}
\frac{\sigma M\int\frac{p_c p^2}{r^3}{\rm d}r+L^2}{p^2}=\frac{M}{r}+\sigma M\int\frac{p_c}{r^3}{\rm d}r+c.
\end{equation}
But, similarly as in the previous example, there is no simple solution (or solution of any kind) known to the author. Furthermore, the nonlocal quantities $\int\frac{p_c p^2}{r^3}{\rm d}r,\int\frac{p_c}{r^3}{\rm d}r$ cannot be easily interpreted and, more importantly, variables $p,r$ are still present so the equation is not ``solved'' in any coordinate system. 

But we can use the other approach: multiplying (\ref{SSP}) by ${\bf \dot x}^\perp$ we obtain
\begin{align*}
{\bf \ddot x}\cdot {\bf \dot x}^\perp&=-\frac{M}{r^3}{\bf x}\cdot {\bf \dot x}^\perp\\
\kappa\abs{\bf \dot x}^3&=\frac{M}{r^3}p \abs{\bf \dot x}\\
\kappa&=\frac{M}{r^3}\frac{p}{\abs{\bf \dot x}^2}=M\frac{ p^3}{r^3}\frac{1}{\zav{{\bf x}^\perp\cdot {\bf \dot x}}^2}\\
\zav{{\bf x}^\perp\cdot {\bf \dot x}}^2&=M\zav{\frac{p}{r}}^3\rho=M\kappa^\star.
\end{align*}
On the other hand, multiplying (\ref{SSP}) by ${\bf x}^\perp$ yields:
\begin{align*}
{\bf \ddot x}\cdot {\bf x}^\perp&=\frac{\sigma M}{r^2} \zav{\frac{{\bf x}\cdot{\dot{\bf x}}}{\abs{\dot{\bf x}}r}}^2 \frac{\dot{\bf x}\cdot {\bf x}^\perp}{\abs{\dot{\bf x}}}.\\
\partial_t\zav{{\bf \dot x}\cdot {\bf x}^\perp}&=\frac{\sigma M p_c^2}{r^4}p\\
{\bf \dot x}\cdot {\bf x}^\perp\partial_t\zav{{\bf \dot x}\cdot {\bf x}^\perp}&=\frac{\sigma M p_c^2}{r^4}p{\bf \dot x}\cdot {\bf x}^\perp\\
\partial_t\zav{{\bf \dot x}\cdot {\bf x}^\perp}^2&=\frac{2\sigma M p_c^2}{r^4}p^2\abs{\bf \dot x}\\
\zav{{\bf \dot x}\cdot {\bf x}^\perp}^2&=2\sigma M\int\frac{p_c^2 p^2}{r^4}\dd s+c\\
\zav{{\bf \dot x}\cdot {\bf x}^\perp}^2&=2\sigma M\int\frac{p_c p^2}{r^3}\dd r+c.
\end{align*}
Combining these results we obtain
\begin{equation}\label{SSPsol}
M\kappa^\star=2\sigma M\int\frac{p_c p^2}{r^3}\dd r+c. 
\end{equation}


We can also produce an equation with two free parameters rewriting (\ref{SSPsol}) as follows:
$$
\frac{M}{2\sigma M\int\frac{p_c p^2}{r^3}\dd r+c}=\rho^\star,
$$ 
and integrating both sides with respect to $\int\frac{p_c p^2}{r^3}\dd r$:
$$
\frac{1}{2\sigma} \ln \zav{2\sigma M\int\frac{p_c p^2}{r^3}\dd r+c}=\int \frac{r^3}{p^3}\kappa\dd\zav{ \int\frac{p_c p^2}{r^3}\dd r}=\int \frac{r^3}{p^3} \frac{p_c p^2}{r^4}\dd p=\int\frac{p_c}{rp}{\rm d}p+c_2.
$$
Thus we get
\begin{equation}\label{SSPsol2}
2\sigma M\int\frac{p_c p^2}{r^3}\dd r+c=Le^{2\sigma\int \frac{p_c}{rp}\dd p}. 
\end{equation}
\end{example}
\subsection{Notation}
All the nonlocal variables that we have work with so far were in the form
$$
\int \frac{r^\alpha}{p^\beta p_c^\gamma}\dd r,\qquad or\qquad  \int \frac{r^\alpha}{p^\beta p_c^\gamma}\dd p.
$$
Let us formalize this observation:
\begin{definition}
Denote 
\begin{equation}\label{nonlocaldef}
R^\alpha_{\beta,\gamma}:=\int \frac{r^\alpha}{p^\beta p_c^\gamma}\dd r,\qquad  P^\alpha_{\beta,\gamma}:=\int \frac{r^\alpha}{p^\beta p_c^\gamma}\dd p.
\end{equation}
\end{definition}
Examples:
\begin{align*}
A&=\frac12 R^{1}_{-1,1}, & s&=R^{1}_{0,1}, & \varphi&=R^{-1}_{-1,1}, & \theta&=P^{0}_{0,1},\\
A^\star&=\frac12 P^{0}_{2,1}, & s^\star&=P^{1}_{2,1}, & \varphi^\star&=P^{0}_{0,1}, & \theta^\star&=R^{-1}_{-1,1},\\
\int r^\alpha{\rm d}\varphi&=R^{\alpha-1}_{-1,1}, & \int \rho^\star{\rm d}s&=P^{3}_{3,1}, & \int r^\alpha {\rm d}\theta&=P^\alpha_{0,1}.
\end{align*}
Solution to our examples can be written in this notation as follows:
\begin{align}
\text{Aerobraking}& & \ddot {\bf x}&=-\frac{M}{r^3}{\bf x}-\alpha \abs{\dot{\bf x}} \dot{\bf x}, &  \frac{M}{2\alpha}e^{2\alpha R^{1}_{0,1}}+c&=L P^{3}_{3,1}.\\
\text{Aerobraking+}& & \ddot {\bf x}&=-\frac{M}{r^3}{\bf x}-\alpha r^\beta \abs{\dot{\bf x}} \dot{\bf x}, &  \frac{M}{2\alpha}e^{2\alpha R^{\beta+1}_{0,1}}+c&=L P^{\beta+3}_{3,1}.\\
\text{Solar Sail prograde} & &
\ddot {\bf x}&=-\frac{M}{r^3}{\bf x}+\frac{\sigma M}{r^2} \zav{\frac{{\bf x}\cdot{\dot{\bf x}}}{\abs{\dot{\bf x}}r}}^2 \frac{\dot{\bf x}}{\abs{\dot{\bf x}}}, & 2\sigma M R^{-3}_{-2,-1}+c&=Le^{2\sigma P^{-1}_{1,-1}}. 
\end{align}
The point of this notation is that instead of trying to solve these equations in known coordinates, it might be more insightful to try understand the coordinates that a problem suggests.

This approach is definitely useless if brought too far, since \textit{ any} equation can be interpreted as a linear equation in sufficiently nice coordinates -- just denote its right hand side and left hand side as a new coordinates!

But nonlocal coordinates as we have defined them posses some beautiful properties, as we are about to show.

\subsection{Properties}
Nonlocal variables are closed under scaling $S_a$, dual $D$, power transform $M_a$ and pedal transform $P$:
\begin{lemma}
Let $\sigma:=\alpha-\beta-\gamma+1$ denote ``parameter excess'' and let $a>0$. Then the following holds: 
\begin{align*}
R^\alpha_{\beta,\gamma}&\stackrel{S_a}{\longrightarrow}a^{\sigma} R^{\alpha}_{\beta,\gamma}, & P^\alpha_{\beta,\gamma}&\stackrel{S_a}{\longrightarrow}a^{\sigma} P^{\alpha}_{\beta,\gamma}, & \sigma&\stackrel{S_a}{\longrightarrow} \sigma \\
R^\alpha_{\beta,\gamma}&\stackrel{D}{\longrightarrow}(-1)^{\gamma+1} P^{\alpha+1-\sigma}_{\beta+1+\sigma,\gamma}, & P^\alpha_{\beta,\gamma}&\stackrel{D}{\longrightarrow}(-1)^{\gamma+1}R^{\alpha-1-\sigma}_{\beta-1+\sigma,\gamma}, & \sigma&\stackrel{D}{\longrightarrow} -\sigma\\
R^\alpha_{\beta,\gamma}&\stackrel{M_{a}}{\longrightarrow}a R^{\alpha+(a-1)\sigma}_{\beta,\gamma}, & P^\alpha_{\beta,\gamma}&\stackrel{M_{a}}{\longrightarrow}(a-1)R^{\alpha-1+(a-1)\sigma }_{\beta-1,\gamma}+P^{\alpha+(a-1)\sigma}_{\beta,\gamma}, & \sigma&\stackrel{M_a}{\longrightarrow} a\sigma\\
R^\alpha_{\beta,\gamma}&\stackrel{P}{\longrightarrow}2 R^{\alpha+\sigma}_{\beta+\sigma,\gamma}- P^{\alpha+1+\sigma}_{\beta+1+\sigma,\gamma}, & P^\alpha_{\beta,\gamma}&\stackrel{P}{\longrightarrow}  R^{\alpha-1+\sigma}_{\beta-1+\sigma,\gamma}-2 P^{\alpha+\sigma}_{\beta+\sigma,\gamma} & \sigma&\stackrel{P}{\longrightarrow} \sigma.
\end{align*}
\end{lemma}
\begin{proof} The proof is an easy exercise and it is left to the reader. 
\end{proof}
The third column shows how various transform affect the parameter excess $\sigma$. Notice that the parameter $\gamma$ is not affected at all as though nonlocal variables with different $\gamma$ cannot be transformed to one another.
 
But since $p_c^2=r^2-p^2$ we have
\begin{equation}\label{nonlocconection}
R^\alpha_{\beta,\gamma-2}=R^{\alpha+2}_{\beta,\gamma}-R^{\alpha}_{\beta-2,\gamma},\qquad
P^\alpha_{\beta,\gamma-2}=P^{\alpha+2}_{\beta,\gamma}-P^{\alpha}_{\beta-2,\gamma}.
\end{equation}
There is also a connection formula for nonlocal variables with the same $\gamma$:
\begin{lemma}
It holds:
$$
(\alpha+1-\gamma)R^{\alpha}_{\beta,\gamma}-(\alpha-1)R^{\alpha-2}_{\beta-2,\gamma}=\frac{r^{\alpha-1}}{p^{\beta}p_c^{\gamma-2}}+\beta P^{\alpha+1}_{\beta+1,\gamma}-(\gamma+\beta-2)P^{\alpha-1}_{\beta-1,\gamma}+c,
$$
where $c\equiv c(\alpha,\beta,\gamma)$ is an integration constant generally dependent on the values of $\alpha,\beta,\gamma$. 

In particular we have:
\begin{align}
R^{1}_{1,1}&=\frac{p_c}{p}+P^{2}_{2,1}, & \alpha&=1,\quad \beta=1,\quad \gamma=1.\\
s&=p_c+P^{0}_{-1,1}, & \alpha&=1,\quad \beta=0,\quad \gamma=1.\\
2A&=pp_c-P^2_{0,1}+2P^{0}_{-2,1}, & \alpha&=1,\quad \beta=-1,\quad \gamma=1.\\
\psi&=\theta-\varphi=\frac{p}{r^2}p_c-2\zav{R^{-3}_{-3,1}-P^{-2}_{-2,1}}, & \alpha&=-1,\quad \beta=-1,\quad \gamma=1.\\
s^\star&=p^\star_c+R^{-2}_{-1,1}, & \alpha&=0,\quad \beta=1,\quad \gamma=1.\\
R^{-2}_{-2,1}&=\frac{p_c}{r}+P^{-1}_{-1,1}, & \alpha&=0,\quad \beta=0,\quad \gamma=1.
\end{align}
\end{lemma}
\begin{proof}
This is, essentially, the result of  an integration by parts. Start with
\begin{align*}
(\alpha+1)R^{\alpha}_{\beta,\gamma}&=\int \frac{(\alpha+1)r^{\alpha}}{p^\beta p_c^\gamma}\dd r=\int \frac{(\alpha+1)r^{\alpha} r'_p}{p^\beta p_c^\gamma}\dd p= \frac{r^{\alpha+1}}{p^\beta p_c^\gamma}-\int r^{\alpha}\zav{\frac{1}{p^\beta p_c^\gamma}}'_p\dd p\\
&=\frac{r^{\alpha+1}}{p^\beta p_c^\gamma}-\int r^{\alpha}\zav{-\frac{\beta}{p}-\frac{\gamma p_c'}{p_c}}\frac{r^{\alpha+1}}{p^\beta p_c^\gamma}\dd p\\
\intertext{using $p_c{p_c}'p=rr'-p$ we get}
&=\frac{r^{\alpha+1}}{p^\beta p_c^\gamma}-\int r^{\alpha}\zav{-\frac{\beta}{p}-\frac{\gamma rr'-p}{p^2_c}}\frac{r^{\alpha+1}}{p^\beta p_c^\gamma}\dd p\\
&=\frac{r^{\alpha+1}}{p^\beta p_c^\gamma}+\gamma \int\frac{r^{\alpha+2}}{p^\beta p_c^{\gamma+2}}\dd r -(\gamma+\beta)\int\frac{r^{\alpha+1}}{p^{\beta-1} p_c^{\gamma+2}}\dd p+\beta \int\frac{r^{\alpha+3}}{p^{\beta+1} p_c^{\gamma+2}}\dd p\\
&=\frac{r^{\alpha+1}}{p^\beta p_c^\gamma}+\gamma R^{\alpha+2}_{\beta,\gamma+2}-(\gamma+\beta)P^{\alpha+1}_{\beta-1,\gamma+2}+\beta P^{\alpha+3}_{\beta+1,\gamma+2}.
\end{align*}
Applying the identity (\ref{nonlocconection}) we can transform the factor $(\alpha+1)R^{\alpha}_{\beta,\gamma}$ so that it has the same $\gamma+2$ parameter as the rest of the non-local variables. Renaming the parameters finishes the proof.
\end{proof}

\begin{example} (Solar sail keeping constant angle with the Sun) For a solar sail which is moving so that the angle of incidence of the incoming light is kept the same, the equation of motion becomes:
\begin{equation}\label{CA}
\ddot {\bf x}=-\frac{M_1}{r^3}{\bf x}+\frac{M_2}{r^3}{\bf x}^\perp.
\end{equation}
The solar sail does not have to be even fully reflective in this case.  Scattering and absorption of some incoming light will produce only different ratio of $M_1,M_2$. This case was extensively studied in e.g. (\cite{solarsaillog1,solarsaillog2,solarsaillog3}). No conservation law was ever found and there is single known exact solution -- logarithmic spirals.

Using 
$$
{\bf x}^\perp=\frac{p}{p_c} {\bf x}+\frac{r^2}{p_c \abs{\dot {\bf x}}}{\dot {\bf x}}^\perp,
$$ 
we can bring the equation (\ref{CA}) into a form suitable for Theorem \ref{T2} and get:
\begin{equation}\label{CAsolt2}
\frac{M_2\int \frac{p}{p_c}{\rm d}r}{p^2}=\frac{M_1}{r}+M_2\int \frac{p}{r^2 p_c} {\rm d}r.
\end{equation}
of in Nonlocal coordinates:
\begin{equation}\label{CAsolt2nl}
\frac{M_2R^{0}_{-1,1}}{p^2}=\frac{M_1}{r}+M_2 P^{-2}_{-1,1}.
\end{equation}
 
Again, this is in no sense a solution, since there are, as of now, 4 variables -- or this can be seen as a second degree differential equation with no general solution (known to the author). But we can extract form this the logarithmic spiral solution. 

Remember that pedal equation of logarithmic spiral in the form $r=r_0 e^{\alpha \varphi}$ is $ p_c=\alpha p$ or $r^2=(1+\alpha^2)p^2$ (see Example \ref{Logspiralex}).

Substituting into (\ref{CAsolt2}) we get
$$
\frac{\alpha^2+1}{r^2}\zav{M_2 \frac{r}{\alpha}+c_1}=\frac{M_1}{r}-M_2\frac{1}{r\alpha}+c_2,
$$ 
where $c_1,c_2$ are integration constants.

We can see that solution exists only if $c_1=c_2=0$. That means that only if the initial conditions are right we obtain an algebraic equation on $\alpha$:
$$
M_2\alpha^2-M_1\alpha+2M_2=0,
$$ 
with two possible solutions:
$$
\alpha_{\pm}=\frac{M_1\pm\sqrt{M_1^2-8M_2^2}}{2M_2},\qquad M_2\not=0.
$$
For $M_2=0$ we are in Keplerian case and no logarithmic solution exists.
\bigskip

Going back to the general case. Multiplying (\ref{CA}) by ${\bf \dot x}^\perp$ we obtain
\begin{align*}
{\bf \ddot x}\cdot {\bf \dot x}^\perp&=-\frac{M_1}{r^3}{\bf x}\cdot {\bf \dot x}^\perp+\frac{M_2}{r^3}{\bf x}^\perp\cdot {\bf \dot x}^\perp\\
\kappa\abs{\bf \dot x}^3&=\zav{\frac{M_1}{r^3}p+\frac{M_2}{r^3}p_c} \abs{\bf \dot x}\\
\kappa&=\zav{\frac{M_1}{r^3}p+\frac{M_2}{r^3}p_c} \frac{1}{\abs{\bf \dot x}^2}\\
\zav{{\bf x}^\perp\cdot{\bf \dot x}}^2 &=\frac{M_1}{r^3}p^3\rho+\frac{M_2}{r^3}p_c p^2 \rho\\
\zav{{\bf x}^\perp\cdot{\bf \dot x}}^2 &=\zav{M_1+M_2\frac{p_c}{p}}\kappa^\star.
\end{align*}
Multiplying (\ref{CA}) by ${\bf x}^\perp$ yields:
\begin{align*}
{\bf \ddot x}\cdot {\bf x}^\perp&=\frac{M_2}{r}\\
\partial_t\zav{{\bf \dot x}\cdot {\bf x}^\perp}&=\frac{M_2}{r}\\
{\bf \dot x}\cdot {\bf x}^\perp\partial_t\zav{{\bf \dot x}\cdot {\bf x}^\perp}&=\frac{M_2}{r} {\bf \dot x}\cdot {\bf x}^\perp\\
\partial_t\zav{{\bf \dot x}\cdot {\bf x}^\perp}^2&=\frac{2M_2 p}{r}\abs{\bf \dot x}\\
\zav{{\bf \dot x}\cdot {\bf x}^\perp}^2&=2M_2\int \frac{p}{r}\dd s=2M_2\int \frac{p}{ p_c}\dd r+c=2M_2 P^{0}_{-1,1}+c.
\end{align*}
In total:
\begin{equation}\label{CAsol}
\zav{M_1+M_2\frac{p_c}{p}}\kappa^\star=2M_2 P^{0}_{-1,1}+c.
\end{equation}

Unless we wan to call the whole left hand side of this equation a new variable, there is still too many variables. The same approach as before will fix it:
\begin{align*}
\frac{1}{2M_2P^{0}_{-1,1}+c}&=\frac{\rho^\star}{\zav{M_1+M_2\frac{p_c}{p}}}\\
\intertext{Integrating with respect to $ P^{0}_{-1,1}:$}
\frac{1}{2M_2}\ln\zav{2M_2P^{0}_{-1,1}+c}&=\int \frac{\rho^\star}{\zav{M_1+M_2\frac{p_c}{p}}} \dd\zav{P^{0}_{-1,1}}=
\int \frac{r^2}{p^3\zav{M_1+M_2\frac{p_c}{p}}}\frac{p}{ p_c}\dd p\\
\frac{1}{2M_2}\ln\zav{2M_2P^{0}_{-1,1}+c}&=\int \frac{r^2}{p p_c\zav{M_1p+M_2p_c}}\dd p+c_2.
\end{align*} 
Thus we get
\begin{equation}\label{CAsol2}
2M_2P^{0}_{-1,1}+c=Le^{2M_2\int \frac{r^2}{p p_c\zav{M_1p+M_2p_c}}\dd p}.
\end{equation}
As we can see, the problem of Solar sail keeping constant angle with the Sun cannot be solved even in nonlocal coordinates.
\end{example}

\section{Acknowledgement}
The author would like to thank Filip Blaschke and Martin Blaschke for careful reading of the manuscript and suggesting numerous improvements.


\begin{thebibliography}{1}
\bibitem{Blaschke6} P. Blaschke: Pedal coordinates, Dark Kepler and other force problems, Journal of Mathematical Physics, DOI: 10.1063/1.4984905.

\bibitem{Edwards} J. Edwards (1892). \emph{Differential Calculus}. London: MacMillan and Co. pp. 161 ff.
\bibitem{Lawrence}Lawrence, J. D. \emph{A Catalog of Special Plane Curves.} New York: Dover, 1972.

\bibitem{Bell1} Lynden-Bell, D; Lynden-Bell RM (1997). "On the Shapes of Newton's Revolving Orbits". Notes and Records of the Royal Society of London. 51 (2): 195–198. doi:10.1098/rsnr.1997.0016
\bibitem{Bell2} Lynden-Bell D, Jin S (2008). "Analytic central orbits and their transformation group". Monthly Notices of the Royal Astronomical Society. 386 (1): 245–260. arXiv:0711.3491Freely accessible. Bibcode:2008MNRAS.386..245L. doi:10.1111/j.1365-2966.2008.13018.x

\bibitem{Mahomed}Mahomed FM, Vawda F (2000). "Application of Symmetries to Central Force Problems". Nonlinear Dynamics. 21 (4): 307–315. doi:10.1023/A:1008317327402
\bibitem{Newton} Newton I (1999). \emph{The Principia: Mathematical Principles of Natural Philosophy} (3rd edition (1726); translated by I. Bernard Cohen and Anne Whitman, assisted by Julia Budenz ed.). Berkeley, CA: University of California Press. pp. 147–148, 246–264, 534–545. ISBN 978-0-520-08816-0.
\bibitem{williamson} Benjamin Williamson (1899). \emph{An elementary treatise on the differential calculus.} Logmans, Green, and Co.
\bibitem{Yates} R.C. Yates (1952). "Pedal Equations".\emph{A Handbook on Curves and Their Properties}. Ann Arbor, MI: J. W. Edwards.
\bibitem{Zwikker} Zwikker, C. \emph{The Advanced Geometry of Plane Curves and Their Applications.} New York: Dover, 1963.
\bibitem{Zubrin} Zubrin, R. \emph{The Dipole Drive: A New Concept in Space Propulsion}, \url{https://doi.org/10.2514/6.2019-1122}
\bibitem{solarsail} McInnes, C. R. and Brown, J. C. (1989) \textit{Solar Sail Dynamics with an Extended Source of Radiation Pressure}, International Astronautical Federation, IAF-89-350, October.
\bibitem{solarwind}S.J. Schwartz and E. Marsch "The radial evolution of a single solar wind plasma parcel," J. Geophys. Res. 88(A12), pp. 9919-9932, doi:10.1029/JA088iA12p09919, 1983.
\bibitem{solarsaillog1} Tsu, T. C., 1959. Interplanetary travel by solar sail. ARS Journal 29, 422–427 
\bibitem{solarsaillog2} Bacon, R. H., 1959. \textit{Logarithmic spiral:  An ideal trajectory for the interplanetary vehicle with engines of low sustained thrust}. American Journal of Physics 27 (3), 164–165, doi:  10.1119/1.1934788. 
\bibitem{solarsaillog3} M. Bassetto, L. Niccolai, A. A. Quarta, and G. Mengali, “Logarithmic  spiral  trajectories  generated  by  Solar sails,” Celest. Mech. Dyn. Astr. vol. 130, no. 18, pp. 1-24, Feb. 2018, doi:10.1007/s10569-017-9812-6.  
\end{thebibliography}
\end{document}